%% file: root.tex
\documentclass[11pt,twoside,a4paper]{article}
\usepackage{graphicx} 
\usepackage{subfigure}
\usepackage{amsmath,amssymb}
\usepackage{todonotes}
\usepackage{mathrsfs}
\usepackage{xspace}
\usepackage{microtype}
\usepackage{color}
\usepackage{stackengine}
\usepackage{overpic}
\usepackage{varwidth}
\usepackage[hidelinks]{hyperref}
\usepackage{amsthm}

\usepackage[mathlines,switch]{lineno}
\newcommand*\patchAmsMathEnvironmentForLineno[1]{%
	\expandafter\let\csname old#1\expandafter\endcsname\csname #1\endcsname
	\expandafter\let\csname oldend#1\expandafter\endcsname\csname end#1\endcsname
	\renewenvironment{#1}%
	{\linenomath\csname old#1\endcsname}%
	{\csname oldend#1\endcsname\endlinenomath}}%
\newcommand*\patchBothAmsMathEnvironmentsForLineno[1]{%
	\patchAmsMathEnvironmentForLineno{#1}%
	\patchAmsMathEnvironmentForLineno{#1*}}%
\AtBeginDocument{%
	\patchBothAmsMathEnvironmentsForLineno{equation}%
	\patchBothAmsMathEnvironmentsForLineno{align}%
	\patchBothAmsMathEnvironmentsForLineno{flalign}%
	\patchBothAmsMathEnvironmentsForLineno{alignat}%
	\patchBothAmsMathEnvironmentsForLineno{gather}%
	\patchBothAmsMathEnvironmentsForLineno{multline}%
}
\newcommand{\decomp}{\texttt{d}}
\newcommand{\xmonotone}{column convex}
\newcommand{\ymonotone}{row convex}
\newcommand{\envup}{\ensuremath{P_{\texttt{up}}}}
\newcommand{\envlo}{\ensuremath{P_{\texttt{lo}}}}
\newcommand{\rem}[1]{}
\newcommand{\revision}[1]{{\color{black}#1}}
\newtheorem{lemma}{Lemma}
\newtheorem{theorem}{Theorem}
\newtheorem{definition}{Definition}

\renewcommand{\todo}[1]{{}}%comment out to see TODO comments.

\title{Efficient Parallel Self-Assembly\\ Under Uniform Control Inputs} %Use for final RAL version
\author{Arne Schmidt\thanks{Department of Computer Science, TU Braunschweig, Germany. $\{$s.fekete, arne.schmidt$\}$@tu-bs.de},  
Sheryl Manzoor\thanks{Work from these authors was partially supported by National Science Foundation IIS-1553063 and IIS-1619278. Department of Electrical and Computer Engineering, University of Houston, USA. atbecker@uh.edu}, Li Huang$^\dagger$, \\ Aaron T.~Becker$^ \dagger $, and S\'{a}ndor P. Fekete$^*$%	
}
	
\hypersetup{
  pdfauthor = {Arne Schmidt,  Sheryl Manzoor, Li Huang, Aaron T.~Becker, and S\'{a}ndor P. Fekete},
  pdfkeywords = {Computational Geometry, Underactuated Robots, Additive Manufacturing },
  pdftitle = {Efficient Parallel Self-Assembly Under Uniform Control Inputs},
  pdfsubject = {We prove that by successively combining subassemblies, we can achieve sublinear construction times for "staged" assembly of micro-scale  objects from a large number of tiny particles, for vast classes of shapes; this is a significant advance in the context of programmable matter and self-assembly for building high-yield micro-factories.The underlying model has particles moving under the influence of uniform external forces until they hit an obstacle; particles bond when forced together with a compatible particle. Previous work considered sequential composition of objects, resulting in construction time that is linear in the number N of particles, which is inefficient for large N. Our progress implies critical speedup for constructible shapes; for convex polyominoes, even a constant construction time is possible. We also show that our construction process can be used for pipelining, resulting in an amortized constant production time.}
}

\begin{document}
	\maketitle
	\begin{abstract}
				\revision{We prove that by successively combining subassemblies, we can achieve {\em sublinear}
				construction times for ``staged'' assembly of micro-scale
                                objects from a large number of tiny particles, for vast 
				classes of shapes; this is a significant advance 
				in the context of programmable matter and self-assembly for 
				building high-yield micro-factories.
				The underlying model has particles
				moving under the influence of uniform external forces until they hit an obstacle; particles
				bond when forced together with a compatible particle. 
				Previous work considered sequential composition of objects, resulting in construction time 
				that is {\em linear} in the number $N$ of particles, which is inefficient for large $N$. 
				Our progress implies critical speedup for 
				constructible shapes; for convex polyominoes, even a {\em constant} construction 
				time is possible. We also show that our construction process can be used for pipelining, 
				resulting in an {\em amortized constant} production time.}
	\end{abstract}
	% Keywords appear just beneath the abstract. Use only for final RAL version.
	% \begin{IEEEkeywords} Computational Geometry, Underactuated Robots, Additive Manufacturing \end{IEEEkeywords}  % from http://www.ieee-ras.org/publications/ra-l/ra-letters-information-for-authors/184-publications/ra-letters/629-ra-l-keywords
	
% We prove that by successively combining subassemblies, we can achieve sublinear construction times for "staged" assembly of micro-scale  objects from a large number of tiny particles, for vast classes of shapes; this is a significant advance in the context of programmable matter and self-assembly for building high-yield micro-factories.The underlying model has particles moving under the influence of uniform external forces until they hit an obstacle; particles bond when forced together with a compatible particle. Previous work considered sequential composition of objects, resulting in construction time that is linear in the number N of particles, which is inefficient for large N. Our progress implies critical speedup for constructible shapes; for convex polyominoes, even a constant construction time is possible. We also show that our construction process can be used for pipelining, resulting in an amortized constant production time.

\input{01-introduction.tex}

\input{02-preliminaries.tex}
	\input{03-monotone.tex}

\input{04-polyominoes.tex}
	\input{05-demonstration.tex}
	\input{06-conclusion.tex}
	%\bibliography{bibliography}
	\bibliography{IEEEabrv,bibliography}
\end{document}

%% file: 01-introduction.tex
\section{Introduction}
\todo{Some parts of the paper is poorly written, the introduction and
	related works in specific. For instance, fifth sentence in related work
	seems incomplete: "Then, starting with a seed-tile the tiles
	continuously attach to the
	partial assembly."}
%In the field of programmable matter, particles at micro- and nano-scale can be programmed to behave in a certain way.
%For example, DNA can be used to self-assemble arbitrary structures.
The new field of programmable matter gives rise to a wide range of algorithmic questions of geometric flavor.
One of the tasks is designing and running efficient production processes for tiny objects with given shape,
without being able to individually handle the \revision{potentially huge number of particles} from which it 
is composed, e.g., building polyominoes from their tiles without the help of tools.\todo{it does not mention how the
	particles are initially moved in rooms. How to get the input of the
	algorithm is not discussed.}

In this paper we use particles that can be controlled by a uniform external force, 
causing all particles to move in a given direction until they hit an obstacle or another blocked particle, as shown in Fig.~\ref{fig:ortho_convex}.
%By using custom-made obstacle complex rearrangements are possible, even in grid-like environments with axis-parallel motion.
%The appeal of this approach is that it shifts the design complexity from
%the building material (the tiles) to the machinery (the environment). 
 Recent experimental work by Manzoor et al.~\cite{manzoor2017parallel} showed this is practical for simple ``sticky'' particles, enabling
  assembly by sequentially attaching particles \revision{emanating from different depots within the workspace or supply channels from the outside} to the existing subassembly, \revision{as shown in Fig.~\ref{fig:ortho_convex}}.
%By pipelining the production process we may obtain efficient rates, see Fig.~\ref{fig:pipeline}.
The algorithmic challenge is to design the surrounding ``maze'' environment and movement sequence to produce
a desired shape.

A recent paper by Becker et al.~\cite{becker2017tiltassembly} showed that the decision problem of whether a simple polyomino can be built or not
is solvable in polynomial time.  However, this relies on sequential construction
in which one particle at a time is added, resulting in \revision{a {\em linear} number of assembly steps, i.e., a time that 
grows proportional to the number $N$ of particles, which is inefficient for large $N$.} 
\revision{In this paper we provide substantial progress by developing methods that can achieve {\em sublinear}
and in some cases even {\em constant} construction times. Our approaches are based on
hierarchical, ``staged'' processes, in which we allow multi-tile subassemblies to combine at each construction step.}
%With this approach, we are able to achieve the results shown in Table~\ref{tab:results}.
\todo{The introduction requires more supportive context on	motivation and
	the problem statement and it requires better balance between the
	previous work, problem statement and the main contribution section, the
	current version more focused on the contribution and algorithm detail
	other than motivation and problem statement. }

	\begin{figure}
		\centering
		\hfill
		\stackunder[5pt]{
			\label{fig:ortho_convex:i}	\includegraphics[width=0.2\columnwidth]{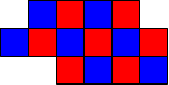}}{Polyomino $P$}\hfill
		 \stackunder[5pt]{
				\label{fig:ortho_convex:a}
				\includegraphics[width=0.3\columnwidth]{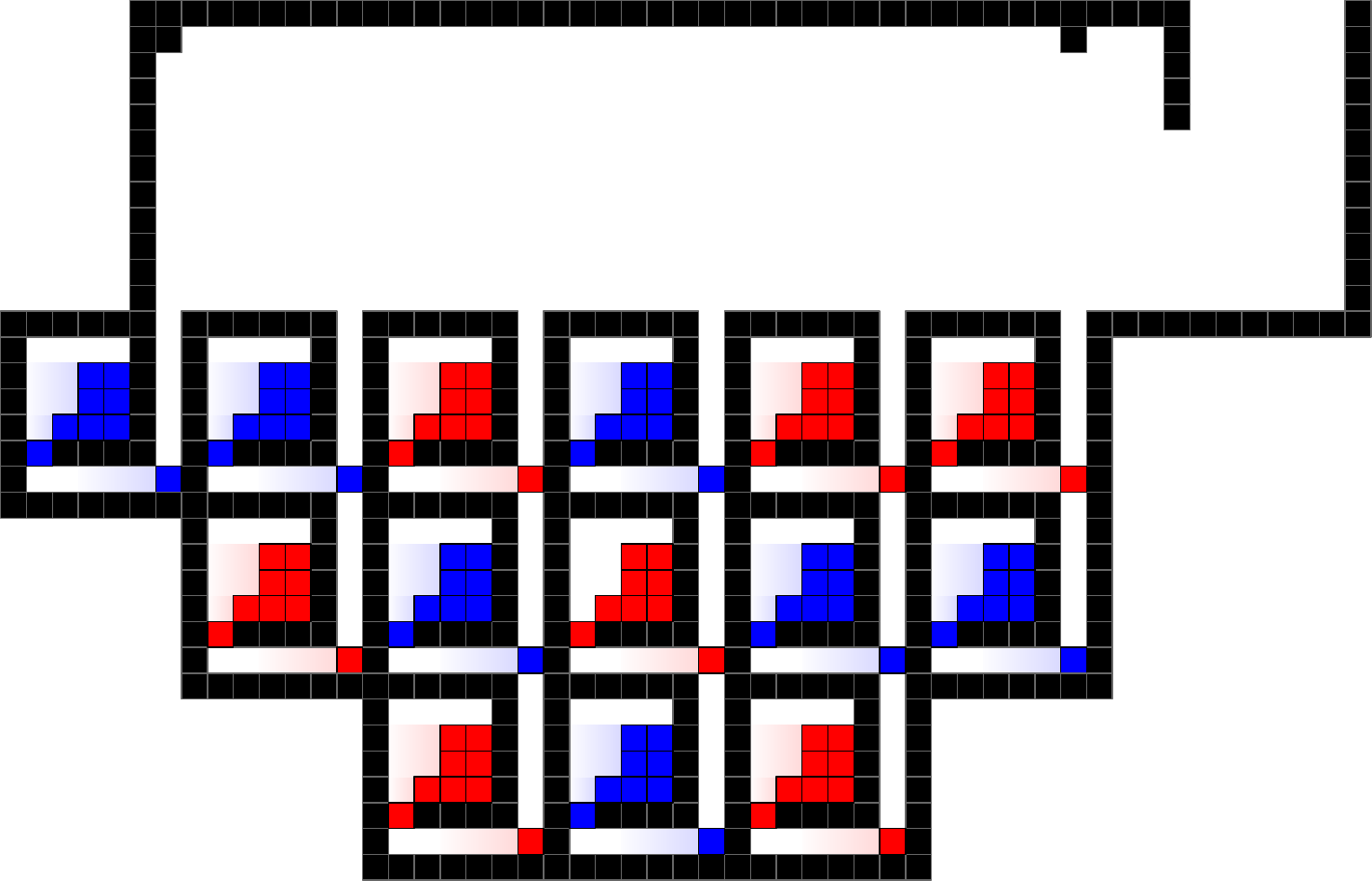}	}{1. Right}		
		 \stackunder[5pt]{
			\label{fig:ortho_convex:b}	
			\includegraphics[width=0.3\columnwidth]{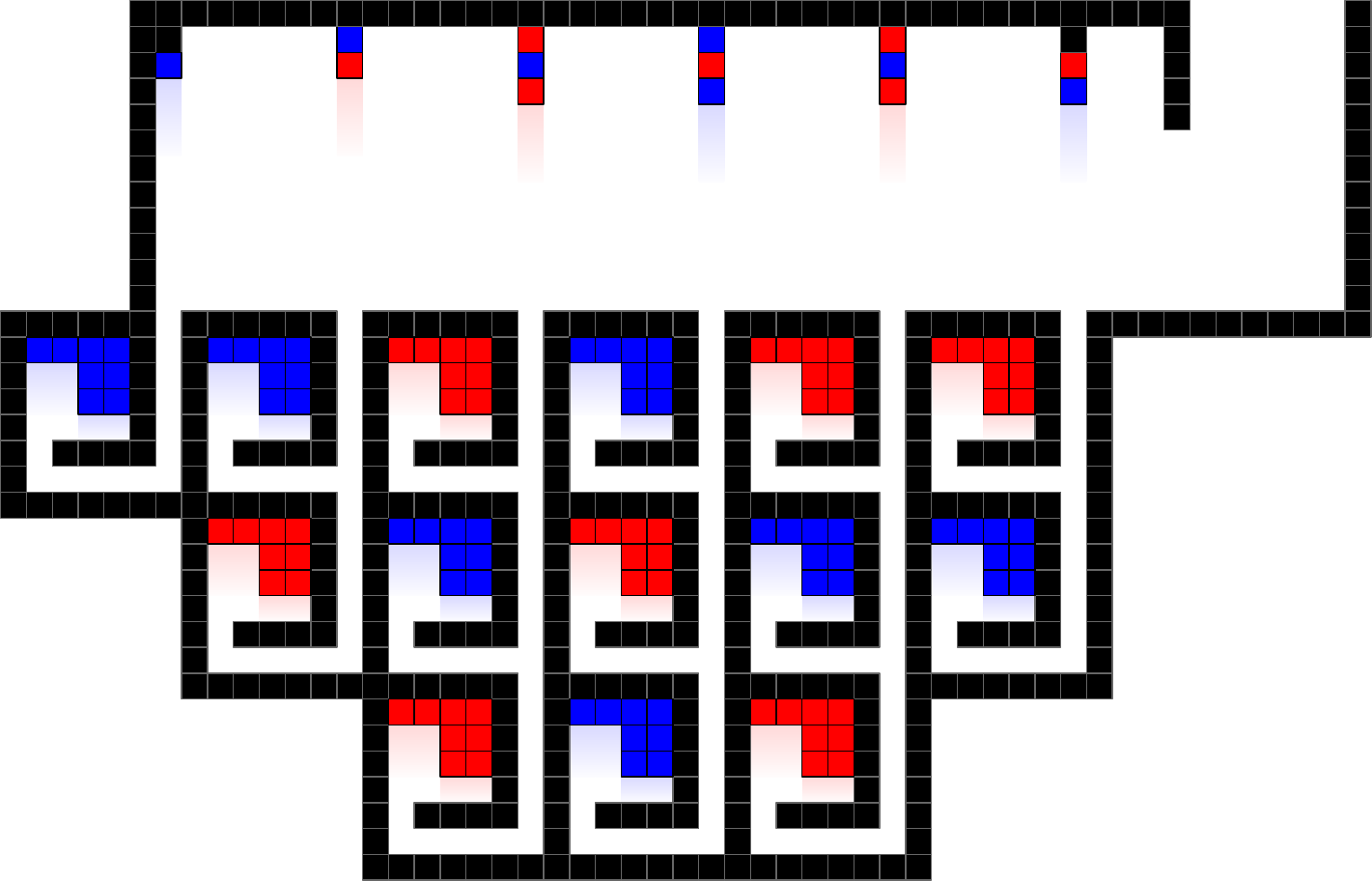}}{2. Up}\\ \vspace{1em}
		
		 \stackunder[5pt]{
			\label{fig:ortho_convex:c}	\includegraphics[width=0.3\columnwidth]{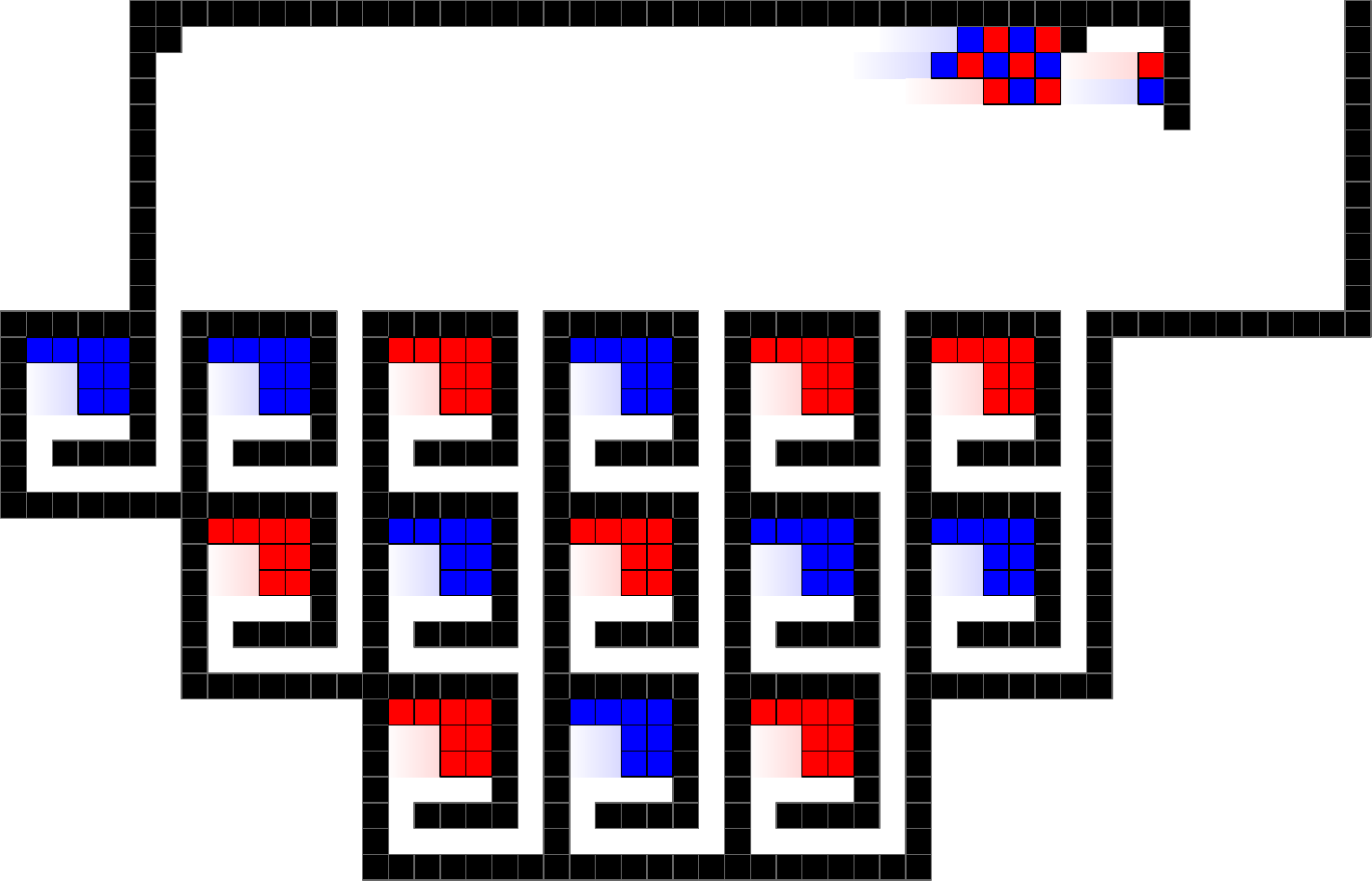}}{3. Right}
		\stackunder[5pt]{
				\label{fig:ortho_convex:d}
				\includegraphics[width=0.3\columnwidth]{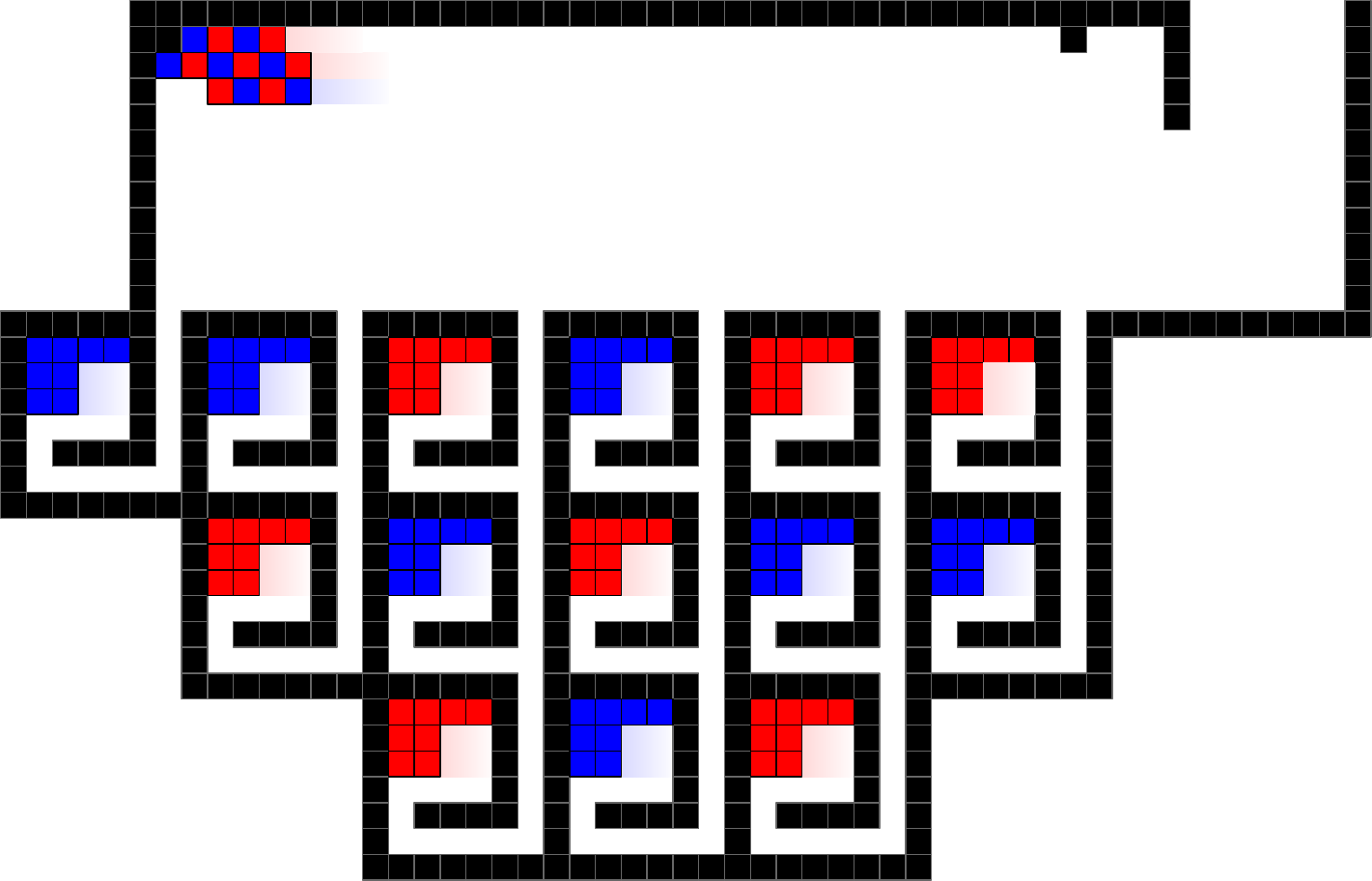}}{4. Left}
		\stackunder[5pt]{
			\label{fig:ortho_convex:e}
			\includegraphics[width=0.3\columnwidth]{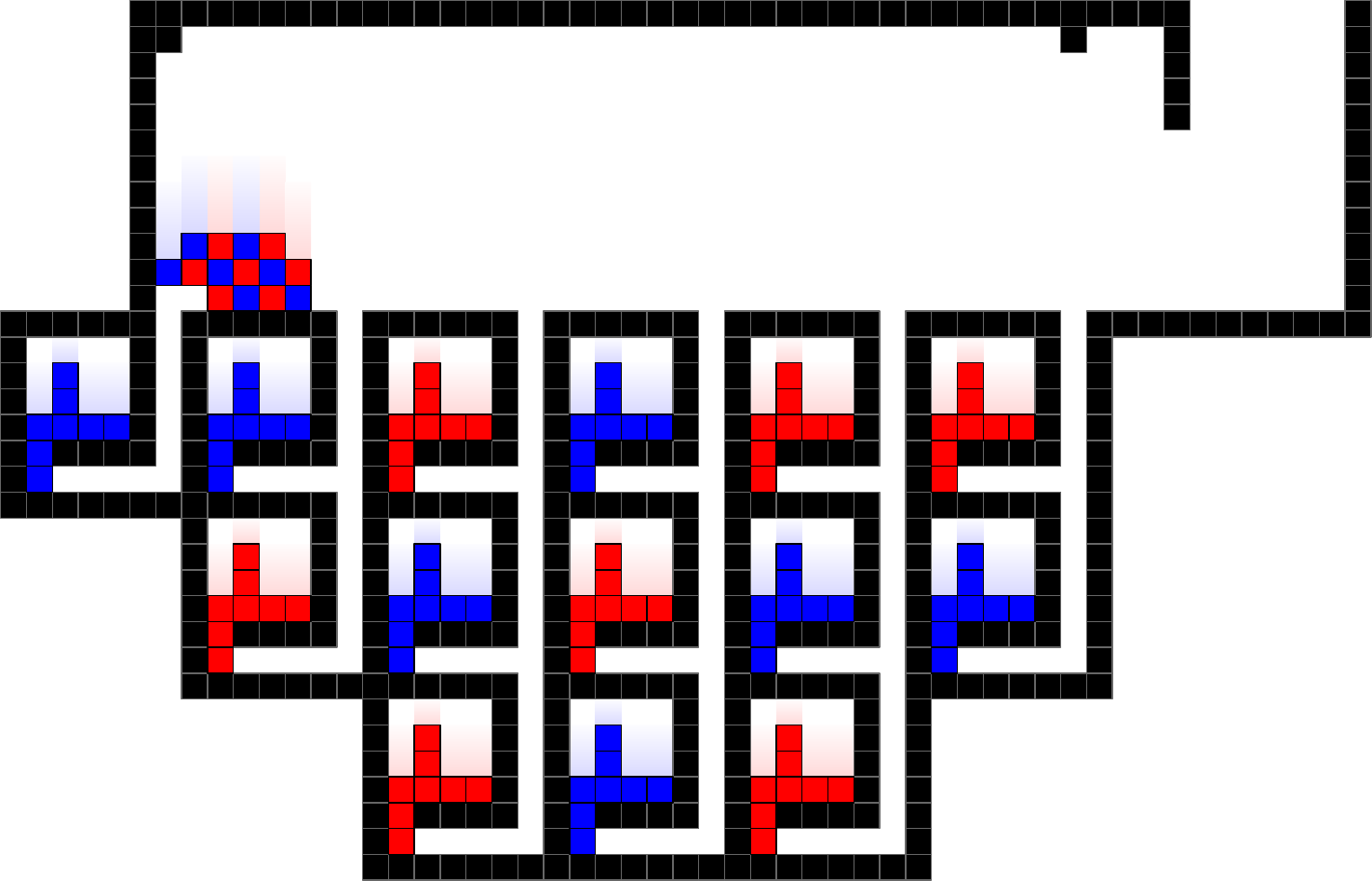}}{5. Down}\\ \vspace{1em}
		
		\stackunder[5pt]{
				\label{fig:ortho_convex:f}
			\includegraphics[width=0.3\columnwidth]{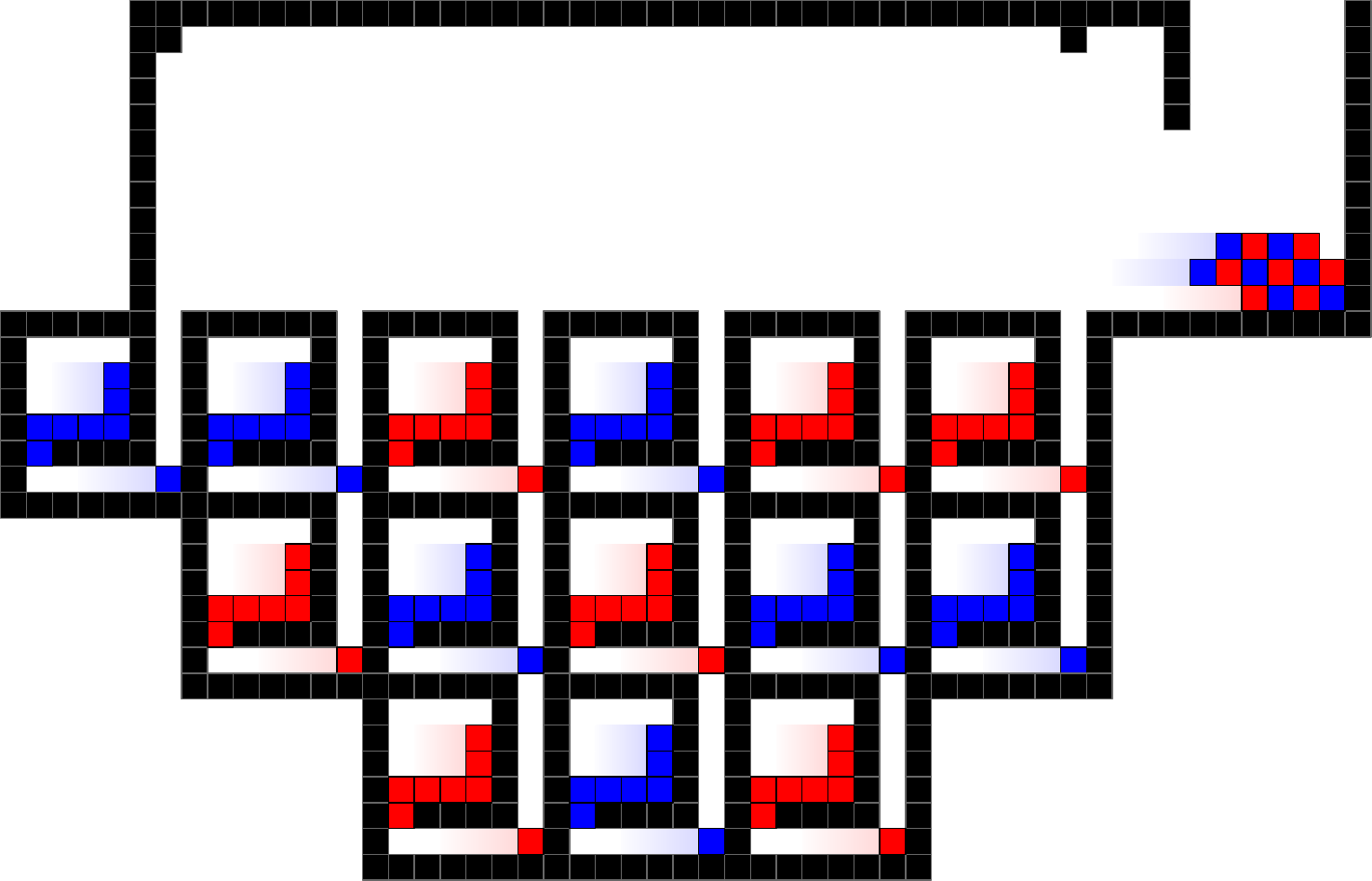}}{6. Right}
		\stackunder[5pt]{
				\label{fig:ortho_convex:g}	
				\includegraphics[width=0.3\columnwidth]{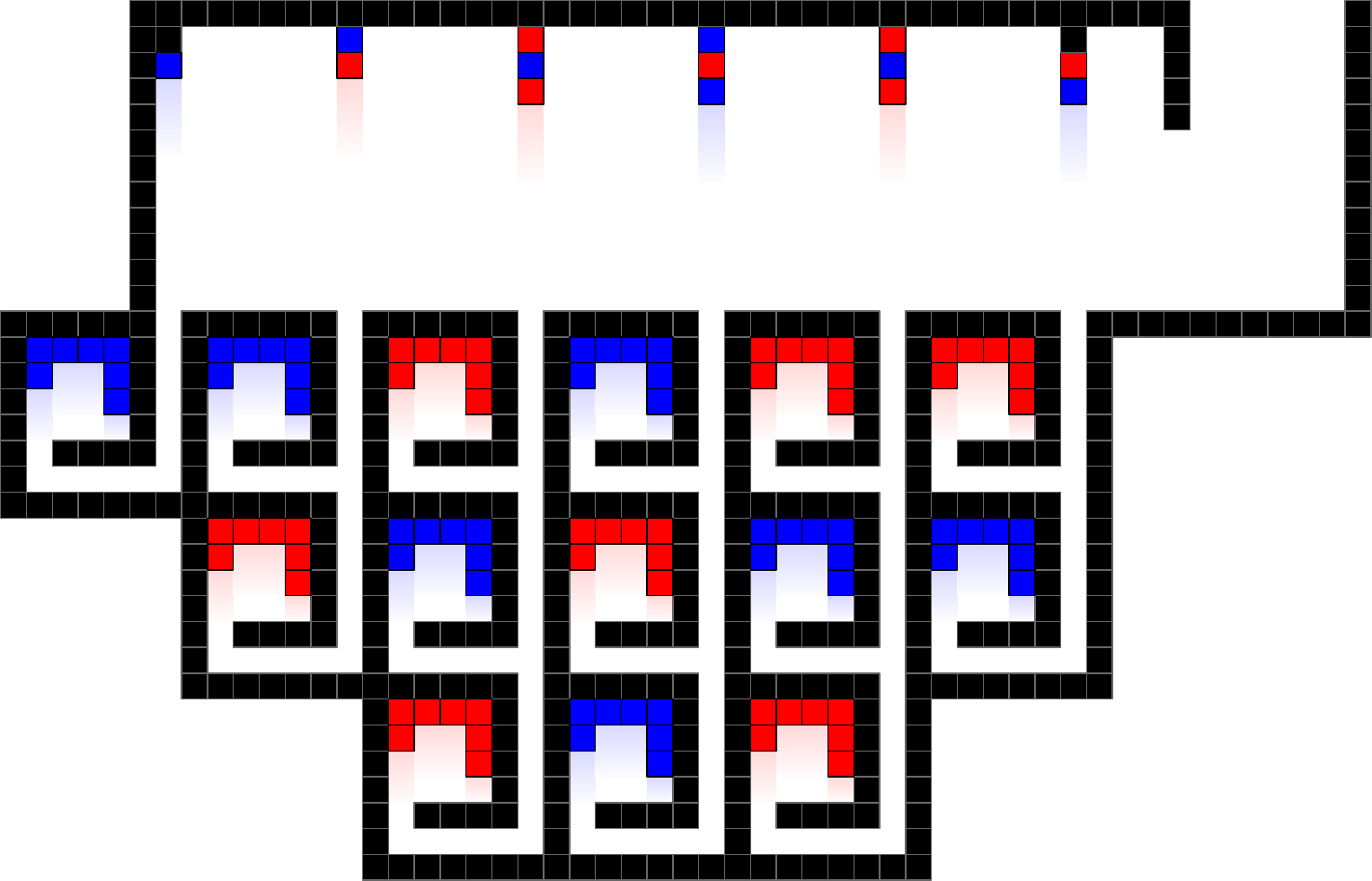}}{7. Up}
			\stackunder[5pt]{
				\label{fig:ortho_convex:h}	\includegraphics[width=0.3\columnwidth]{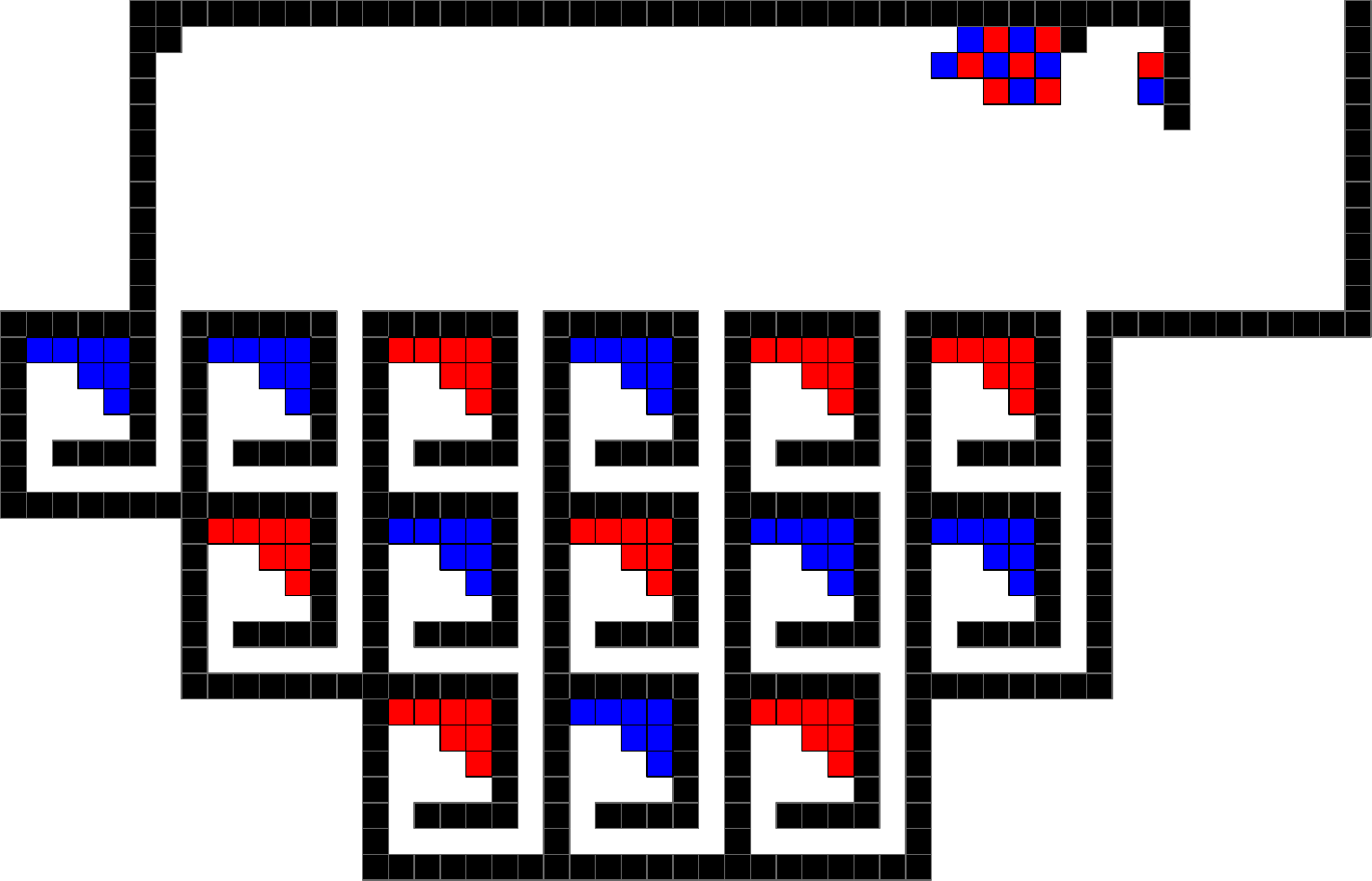}}{8. Right}		
		\caption{Convex polyominoes can be assembled in six movement steps. A copy of the polyomino $P$ is released every five steps after the first copy.
		See video attachment for animation: \protect{\url{https://youtu.be/_R_puO0smPs}}.
		%https://youtu.be/m2lL-uuJa8I, for animation.
		}
		\label{fig:ortho_convex}
		\vspace*{-5mm}
	\end{figure}
\vspace{-0.15cm}
\subsection{Contribution}
\revision{We provide a number of contributions to achieving sublinear construction times 
for polyomino shapes consisting of $N$ pixels (``tiles''), which is critical for the efficient assembly of large
objects. Many of these results are the outcome of decomposing the shape into simpler pieces; as a consequence,
we can describe the construction time in geometric parameters that may be considerably smaller than $N$.
\begin{itemize}
\item
We show that we can decide if a given polyomino $P$ can be recursively constructed from simple
subpieces that are glued together along simple straight cuts (``2-cuts'') in polynomial time.
The resulting production time depends on the number $r(P)$ of \emph{locally reflex} tiles of $P$, which is bounded by $N$,
but may be much smaller. 
\item We show that building a convex polyomino takes $O(1)$ steps.
\item For a monotone polyomino $P$, we need $O(\log \decomp(P))$ steps, where $\decomp(P)\leq N$ is the number of cuts needed to decompose $P$ into convex subpolyominoes.
\item For polyominoes with convex holes, we show that $O(r)$ steps suffice to build the polyomino. 
\item All methods we describe can be pipelined resulting in an amortized constant construction time.
\end{itemize}

We also elaborate the running time for efficiently {\em computing} aspects of the decomposition, as follows.
Finding cuts for a decomposition needs $O(N)$ time for monotone polyominoes.
Simple polyominoes require $O(N+r^2\log r)$ time to find a straight cut and $O(r^2N\log N)$ time to find an arbitrary cut.
Allowing convex holes increases the time to $O(N+r^3\log r)$ and $O(r^3N\log N)$, respectively.

For all these constructions, we show that $N\cdot(\mathcal{C}_P+\sqrt{D})$ obstacles suffice to construct $D$ copies of an $N$-tile polyomino that requires $\mathcal{C}_P$ steps to build.
}

\subsection{Related Work}
In recent years, the problem of assembling a polyomino has been studied intensively using various theoretical models.
Winfree~\cite{winfree1998algorithmic} introduced the \emph{abstract tile self-assembly model} in which tiles with glues on their side can attach to each other if their glue type matches. Then, starting with a \emph{seed-tile}\revision{,} the tiles continuously attach to the partial assembly. If no further tile can attach, the process stops.
Several years later, Cannon et al.~\cite{cannon2012two} introduced the \emph{2-handed  tile self-assembling model (2HAM)} in which sub-assemblies can attach to each other provided that the sum of glue strengths is at least a threshold $\tau$.
Chen and Doty~\cite{chen2017parallelism} introduced a similar model: the \emph{hierarchical tile self-assembling model}.
In 2008, Demaine et al.~\cite{demaine2008staged} introduced the \emph{staged tile self-assembly model} which is based on the 2HAM. Here, sub-assemblies grow in various bins which can then poured together to gain new assemblies.
This model was then further analyzed by Demaine et al.~\cite{demaine2017new} and Chalk et al.~\cite{chalk2016optimal}.
An interesting aspect in all models is that the third dimension can be used to reach specific positions within partial assemblies.
In our paper however, the challenge is to use two dimensions, i.e., an assembly can only bond to another polyomino if the bonding site is completely visible.

All these models have in common that particles, e.g., DNA-strands, self-assemble to bigger structures.
In this paper, however, the particles can only move by global controls and have one glue type on all four sides.
This concept has been studied in practice using biological cells controlled by magnetic fields, see~\cite{kim2015imparting}. \revision{In addition, see~\cite{arbuckle2010self}.}
 Recent work by Zhang et al.~\cite{8206202}  shows there exists a workspace a constant factor larger than the number of agents that enables complete rearrangement for a rectangle of agents.

A more related paper is the work by Manzoor et al.~\cite{manzoor2017parallel}. 
They assemble polyominoes in a pipelined fashion using global control, i.e., by completing a polyomino after each small control sequence the amortized construction time of a polyomino is constant. 
To find a construction sequence building the polyomino only heuristics are used.
Becker et al.~\cite{becker2017tiltassembly} show that it is possible to decide in polynomial time if a hole-free polyomino can be constructed.
However, both papers consider adding one tile at a time.
In this paper, we allow combining partial assemblies at each step. 
We are also able to pipeline this process to achieve an amortized constant production time.

The complexity of controlling robots using a global control has been studied.
Becker et al.~\cite{Becker3810a} show that it is NP-hard to decide if an initial configuration of a robot swarm in a given environment can be transformed into another configuration by only using global control but becomes more tractable if it is allowed to \textit{design} the environment.
Finding an optimal control sequence is even harder. 
\rem{Becker et al.~\cite{becker2014particle} show this problem is PSPACE-complete.}
Related work for reconfiguration of robots with local movement control include work by Walter et al.~\cite{Walter2004}, Vassilvitskii et al.~\cite{vassilvitskii2002general}, and Butler et al.~\cite{butler2004generic}.

%% file: 02-preliminaries.tex
	\section{Preliminaries}\label{sec:Preliminaries}
	\noindent\textbf{Workspace:} A \emph{workspace} $\mathcal{W}$ is a planar grid filled with unit-square particles and fixed unit square blocks (\emph{obstacles}).
	Each cell of the workspace contains either a particle, an obstacle, or the cell is \emph{free}.
	
	\noindent\textbf{Movement step:} A \emph{movement step} is one of the four directions \emph{up}, \emph{right}, \emph{down}, \emph{left}.
	One movement step forces every tile or assembly to move to the specified direction until the tile/assembly is blocked by an obstacle.
	
	\noindent\textbf{Polyomino:}
	For a set $P\subset \mathbb{Z}^2$ of $N$ grid points in the plane, the graph
	$G_P$ is the induced grid graph, in which two vertices $p_1, p_2\in P$ are
	connected if they are at unit distance. Any set $P$ with connected grid graph
	$G_P$ gives rise to a {\em polyomino} by replacing each point $p\in P$ by a unit 
	square centered at $p$, which is called a {\em tile}; for simplicity, we also use $P$ to denote the
	polyomino when the context is clear, and refer to $G_P$ as the dual graph of
	the polyomino.
	A polyomino is called {\em hole-free} or {\em simple} if and only if the grid graph
	induced by $\mathbb{Z}^2\setminus P$ is connected. 
	A polyomino $P$ is \emph{\xmonotone} (\emph{\ymonotone}, resp.) if the intersection of any vertical (horizontal, resp.) line and $P$ is connected, i.e., the polyomino is $x$-monotone ($y$-monotone, resp.).
	Furthermore, a polyomino $P$ is called \emph{(orthogonal) convex} if $P$ is column and row convex.
	
	\noindent\textbf{Tiles:}
		A tile $t$ is an unit-square of a polyomino and also represent particles in the workspace.
		There are two kinds of tiles: blue and red tiles.
		Two tiles stick together if their color differs.
	%	\revision{A tile $t$ of a polyomino $P$ is said to be \emph{locally convex} if there exists a $2\times 2$ square solely containing $t$.
	%	If the square contains $t$ and two neighbors of $t$ only, then we call $t$ \emph{locally reflex}}.
		%{if there are two parameters $a,b\in\{-1,1\}$ such that the three positions $p_1 = p_t + \left(\begin{array}{c}a\\0
%		\end{array}\right)$, $p_2 = p_t + \left(\begin{array}{c}
%		0\\b
%		\end{array}\right)$, $p_3 = p_t + \left(\begin{array}{c}
%		a\\b
%		\end{array}\right)$ are all empty, where $p_t$ is the position of tile $t$.}
%		Furthermore, $t$ is said to be \emph{locally reflex if $p_3$ is empty, and there is a tile at $p_1$ and $p_2$.
%		A tile $t$ is a \emph{corner tile} if $t$ is locally reflex or locally convex. }
		%Note that a tile can be \revision{locally} convex and \revision{locally} reflex at the same time (see Fig.~\ref{fig:corner_tiles}).
		
		\noindent\textbf{Constructibility:}
		A polyomino $P$ is \emph{constructible} if there exists a workspace $\mathcal{W}$ and a sequence $\sigma$ of movement steps that produce $P$.  
		
\rem{	\begin{figure}
		\centering
		\includegraphics[width=0.25\columnwidth]{./figures/corner_tiles.pdf}
		\caption{A polyomino with \revision{locally} convex tiles (red), \revision{locally} reflex tiles (blue), and tiles that are both \revision{locally} convex and \revision{locally} reflex (orange, striped)}
		\label{fig:corner_tiles}
		\vspace*{-6mm}
	\end{figure}
}

	\noindent\textbf{Cuts:}
	A \emph{cut} is an orthogonal curve moving between points of $\mathbb{Z}^2$.
	If any intersection of a cut with the polyomino $P$ has no turn, the cut is called \emph{straight}.
	A $p$\emph{-cut} is a cut that splits a polyomino $P$ into $p$ subpolyominoes.
	Furthermore, a cut is called \emph{valid} if all induced subpolyominoes can be pulled apart into opposite directions without blocking each other.
	A polyomino $P$ is called \emph{(straight) 2-cuttable} if there is a sequence of valid (straight) 2-cuts that subdivide $P$ into monotone subpolyominoes.
	If the subpolyominoes can be pulled apart in horizontal (vertical) directions, we call the cut \emph{vertical (horizontal)}
	An example for 2-cuts can be seen in Fig.~\ref{fig:2-cut:def}.
	In the following we only consider 2-cuts for non-convex polyominoes.

\begin{figure}
	\centering
	\includegraphics[width=0.4\columnwidth]{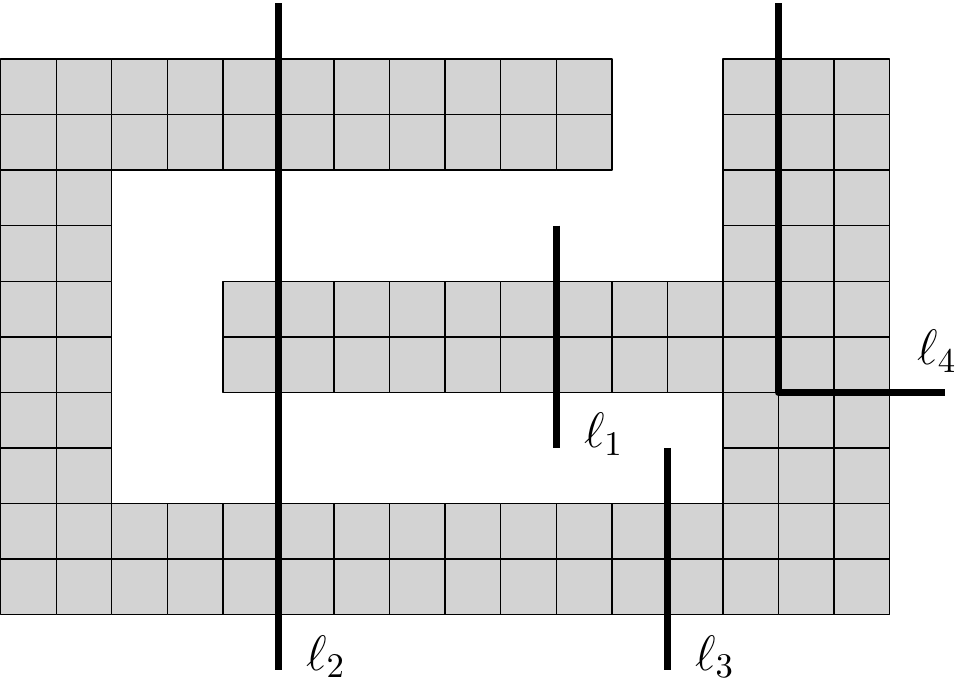}
	\hfill
	\includegraphics[width=0.57\columnwidth]{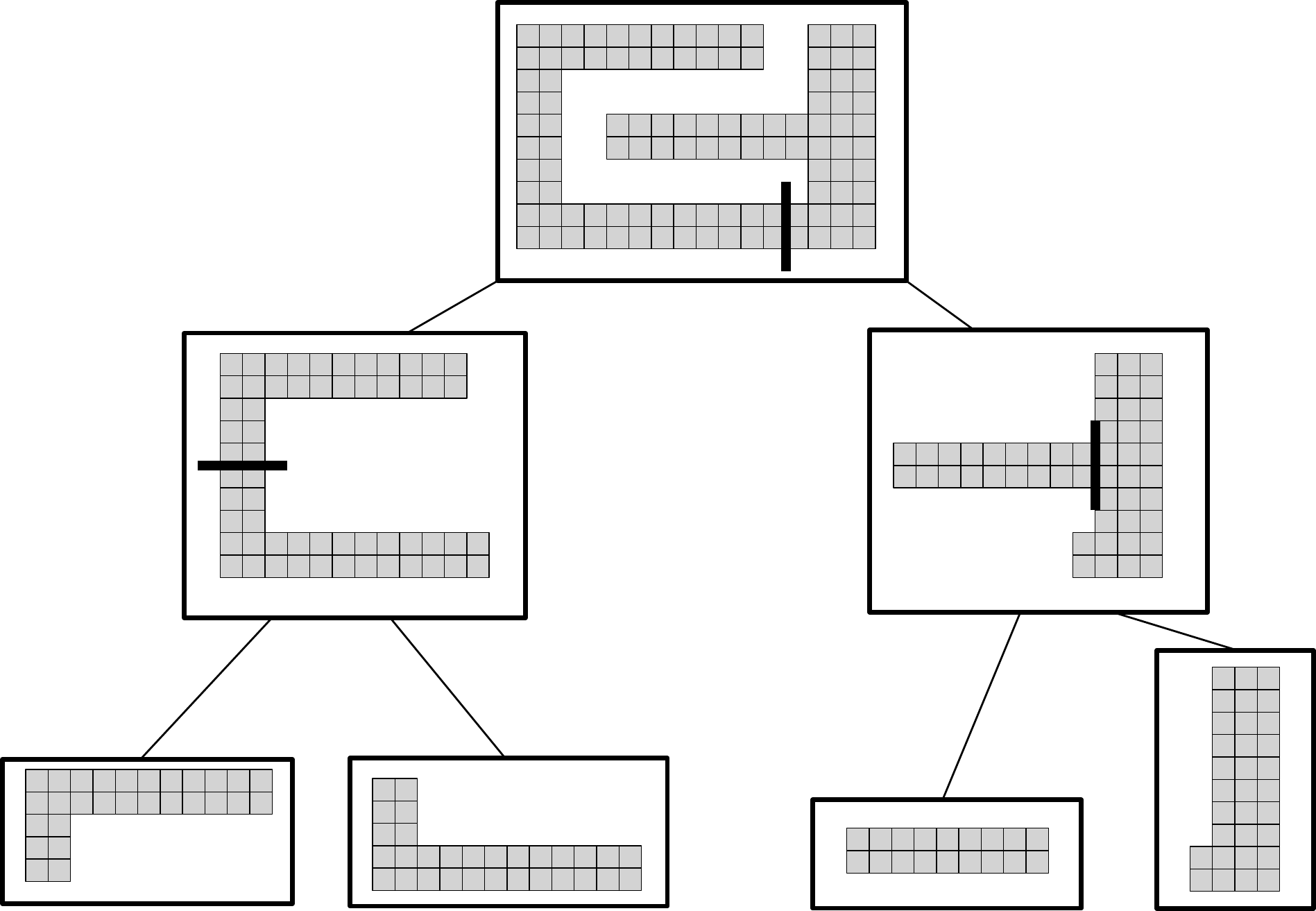}
	\caption{
		Left: 
		(Counter-)Examples for straight 2-cuts: 
		$\ell_1$ is not a 2-cut because we cannot move the left component to the right or left without getting blocked by the other component.
		$\ell_2$ is not a 2-cut because we get more than two components.
		$\ell_3$ is a 2-cut because we get two components which can be pulled apart. $\ell_4$ is not a straight cut. 
		Right: 
		Decomposition tree with straight 2-cuts where the leaves are convex polyominoes.
	}
	\label{fig:2-cut:def}
		\vspace*{-4mm}
\end{figure}

%% file: 03-monotone.tex
	\section{Monotone Assemblies}\label{sec:MonotoneAssemblies}
	This section focuses on convex and monotone polyominoes.

	\begin{lemma}
		\label{lem:ortho_convex}
		Any convex polyomino $P$ can be assembled in six movement steps.
		\todo{The main question about the designing the workspace remains unanswered
			as well as what sequence of the movement is required, do we need to
			have different movement sequences for every polyomino? is there any
			intelligent strategy for that?	Is there any way to automatize it? In
			the proof of lemma 1, right, left, and down movements are used to
			connect columns together, but it is not clear if it is true for
			constructing any convex polyomino, and is it optimal?}
	\end{lemma}
	
	\begin{proof}
		The idea of this proof is simple: 
		Subdivide $P$ into vertical lines of width one, build the lines in two steps (see Fig.~\ref{fig:ortho_convex}.1~and~\ref{fig:ortho_convex}.2), and connect these lines with a right and left movement (see Fig.~\ref{fig:ortho_convex}.3~and~\ref{fig:ortho_convex}.4).
		With two more movements we can flush $P$ out of the labyrinth (see Fig.~\ref{fig:ortho_convex}.5~and~\ref{fig:ortho_convex}.6).
		
		\begin{description}
			\item[Assembling a column:] To construct a column of length $n$, we build $n$ containers, each below the previous.
			Each container releases a new tile after each left, down, right movement combination.
			After the right movement all $n$ tiles move to a wall and then have the same $x$-coordinate.
			With an up movement all $n$ tiles stick to a column once the first tile hits the top wall.
			
			\item[Assembling the polyomino:] Assume we have built each column of the polyomino in parallel. 
			With obstacles we can stop each column at the appropriate respective heights.
			A right movement combines all columns left of the column with the maximum height, and a left movement completes the assembly of $P$.
			To remove the polyomino from the assembly area we use a down, right movement.
			Note that the last three movements are left, down, and right, by which we start the next copy of $P$.
		\end{description}
		%\textcolor{blue}{What are the  'small corridors'?  This following paragraph is awkward.}
		%
		%My main concern is that the current assembly method of a convex polyomino will not work for an L-shaped polyomino (rotated 90 deg clockwise). After the down movement (Fig4(g)) it will get stuck in the leftmost corridor and the next right movement will not bring the assembly to the exit. I don't see an easy solution that would work for all possible shapes which can "get stuck".

%Aaron:  I think the solution to this is to create a cavity to the left of the assembly area such that during the left move in step 4 the entire polyomino moves into the cavity, which has a smooth bottom such that the down move in step 5 does not cause any problems. You do this in Fig. 8.
% Arne, please add obstacles to the top of Fig 8 to stop the completed polyomino during the 'up' move
% Arne, please fill the hoppers with more particles in Fig 9
	
		%
		Without further precautions, a polyomino could get stuck in narrow corridors. 
		This problem can be avoided with a simple case analysis.
		First, observe that the leftmost of the topmost tiles of the polyomino is blocked by an obstacle. 
		Let $t$ be this tile and let $x_t$ be the corresponding $x$-coordinate.
		Also, let $s$ be a tile stuck in a corridor having $x$-coordinate $x_s$.
		Only two cases can occur.
		(a) $x_s < x_t$: We place an additional obstacle directly where $t$ was blocked. This forces the polyomino to stop one position earlier.
		(b) $x_s > x_t$: We shift every obstacle with $x$-coordinate higher than the corridor one unit to the right and we add an additional obstacle at the corridor end. 
		The polyomino is then stopped by this obstacle.		
		%In this case we place every obstacle with higher $x$-coordinate than the corridor one unit to the right and we add an obstacle to the place, where the polyomino would stuck.
	\end{proof}

	\begin{definition}
		Let $P$ be an $x$-monotone polyomino. The decomposition number %\todo[inline]{different name? } 
		\decomp$(P)$ is the minimum number of vertical cuts required to obtain subpolyominoes that are all convex.
		
		The \emph{upper envelope} $\envup\subset P$ consists of (1) all tiles $T$ on the boundary that have no tiles above, and (2) tiles connecting $T$ along the boundary. Analogously define the \emph{lower envelope} $\envlo\subset P$.
		
		We call a straight row $M=\{m_1,\dots,m_k\}\subset \envup$ a \emph{minimum} of $\envup$ if there are two tiles $t_1$ and $t_2$, for which $t_1$ is connected to the top side of $m_1$ and $t_2$ is connected to the top side of $m_k$.
		Analogously define \emph{maximum} for $\envlo$.
	\end{definition}
	
	\begin{figure}
		\centering
		\includegraphics[width=0.5\columnwidth]{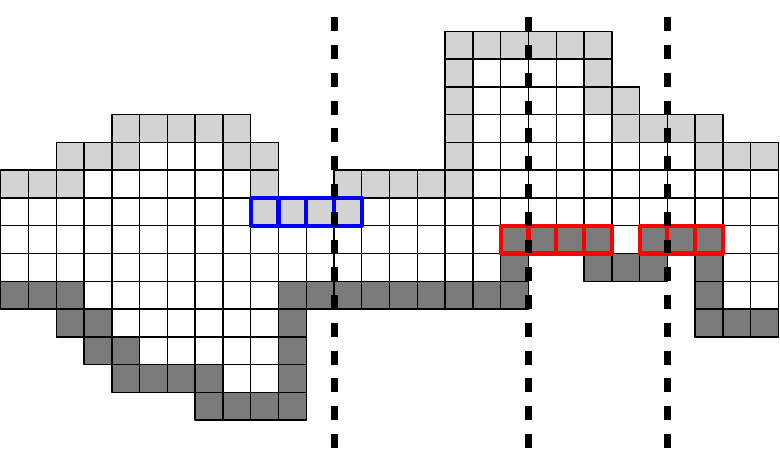}
		\caption{A polyomino $P$. Light grey tiles define $\envup$, dark grey tiles define $\envlo$. Blue framed rows are minima in $\envup$, red framed rows are maxima in $\envlo$. Decomposition number $\decomp(P)=3$, because three vertical lines suffice and three line are necessary because every line hits a maxima/minima.}
		\vspace*{-5mm}
	\end{figure}

	To construct an $x$-monotone polyomino make vertical cuts through the maxima/minima of $\envup$ and $\envlo$, respectively.
	There are at most $\decomp(P)$ many cuts.
	We now can choose a cut, such that on both subpolyominoes $P'$ and $P''$ the decomposition number $\decomp(P')\leq \decomp(P'')\leq \frac 1 2 \decomp(P)$; this can be done with a median search.
	Repeating this procedure on each resulting subpolyomino yields a decomposition tree with depth $\log \decomp(P)$ whose leafs are convex polyominoes.
	
	\begin{lemma}\label{lem:maxima}
		Let $P$ be a polyomino.
		For each minimum and maximum $M$ there must be a vertical cut $\ell$ going through $M$ in order to decompose $P$ into convex subpolyominoes. \todo{lemmas seem to be too simplistic and can potentially be
			better explained in a narrative along with figures}
	\end{lemma}
	
	\begin{proof}
		Suppose we do not need such line $\ell$.
		Let $P'$ be a subpolyomino having a minimum $M'$, through which no cut is made.
		Consider the two tiles $t_1$ and $t_2$ as defined above.
		Both $t_1$ and $t_2$ must be in the same subpolyomino (because there is no cut through $M'$).
		Then, a horizontal line through $t_1$ and $t_2$ enters $P'$ twice and therefore, $P'$ cannot be an convex polyomino. 
	\end{proof}
	
	\begin{lemma}\label{lem:cut_time_monotone}
		Let $P$ be an $x$-monotone polyomino.
		The decomposition number $\decomp(P)$ and the corresponding cuts can be computed in $O(N)$ time.
	\end{lemma}
	
	\begin{proof}
		Finding the minima and maxima of $\envup$ and $\envlo$, respectively, can be found in $O(N)$ time by sweeping from the left boundary to the right boundary.
		Having the minima $M_{\texttt{up}}$ and maxima $M_{\texttt{lo}}$, both in sorted order from left to right, we repeat the following procedure:
		\begin{itemize}
			\item 	Let $M_0\in M_{\texttt{up}}$ and $M'_0\in M_{\texttt{lo}}$ be the leftmost minima/maxima, resp.
			\item If the projection of $M_0$ and $M'_0$ to the $x$-axis overlaps with at least two tiles, then output a vertical line going through $M_0$ and $M'_0$, and remove both from $M_{\texttt{up}}$ and $M_{\texttt{lo}}$, resp.
			\item If this is not the case, output a vertical line going through the minima/maxima that ends first, and remove this minima/maxima from $M_{\texttt{up}}$ or $M_{\texttt{lo}}$, resp.
		\end{itemize}
		This procedure costs $O(\decomp (P))$ time.
		In total, this is $O(N)$ time.
	    The correctness follows from Lemma~\ref{lem:maxima}.
	\end{proof}
	
	\begin{theorem} \label{th:monotone}
		Any $x$-monotone polyomino $P$ with decomposition number \decomp$(P)>0$ can be assembled in $O(\lceil\log (1+\decomp(P))\rceil)$ unit steps.
		Furthermore, this process can be pipelined yielding a construction time of amortized $O(1)$ unit steps.
		%1+\decomp(P) because \decomp(P) might be 0, and log 0 is not good...
		%1+log because log 1 = 0, but we still need 6 steps to build P.
	\end{theorem}
	
	\begin{proof}
		As a first step we search for the vertical cuts as described above.
		Having this subdivision into convex subpolyominoes, we can use Lemma~\ref{lem:ortho_convex} to create all subpolyominoes in parallel.
		We now can use the combining gadget seen in Fig.~\ref{fig:monotone:p1p2} and Fig.~\ref{fig:monotone:p2p1} to combine two adjacent subpolyominoes in each cycle.
		Thus, for each cycle the number of subpolyominoes decreases by a factor of two and we have at most $\lceil\log (1+\decomp(P))\rceil$ cycles to combine all subpolyominoes to obtain $P$. 
		
		As already described in Lemma~\ref{lem:ortho_convex}, we start a new copy after every cycle.
		Thus, to create $D$ copies of $P$ we need $O(\lceil\log (1+\decomp(P))\rceil + D)$ cycles.
		This is an amortized constant time per copy if we create $\Omega(\log\decomp(P))$ copies. Note that $\decomp(P)$ is in $\Omega(1)$ and $O(N)$.
	\end{proof}
	
	\begin{figure}[]
		\centering
		\subfigure[Up]{
				\includegraphics[scale=0.235]{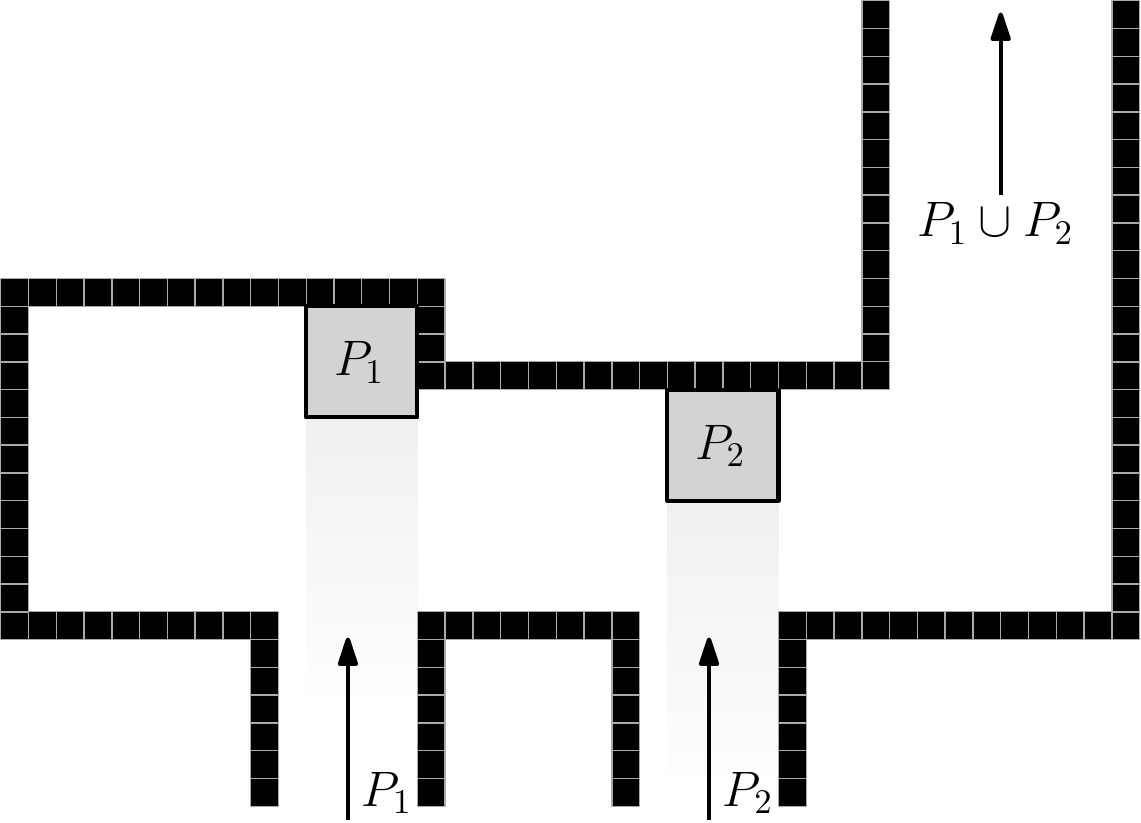}}
			\hfill
		\subfigure[Right]{
			\includegraphics[scale=0.235]{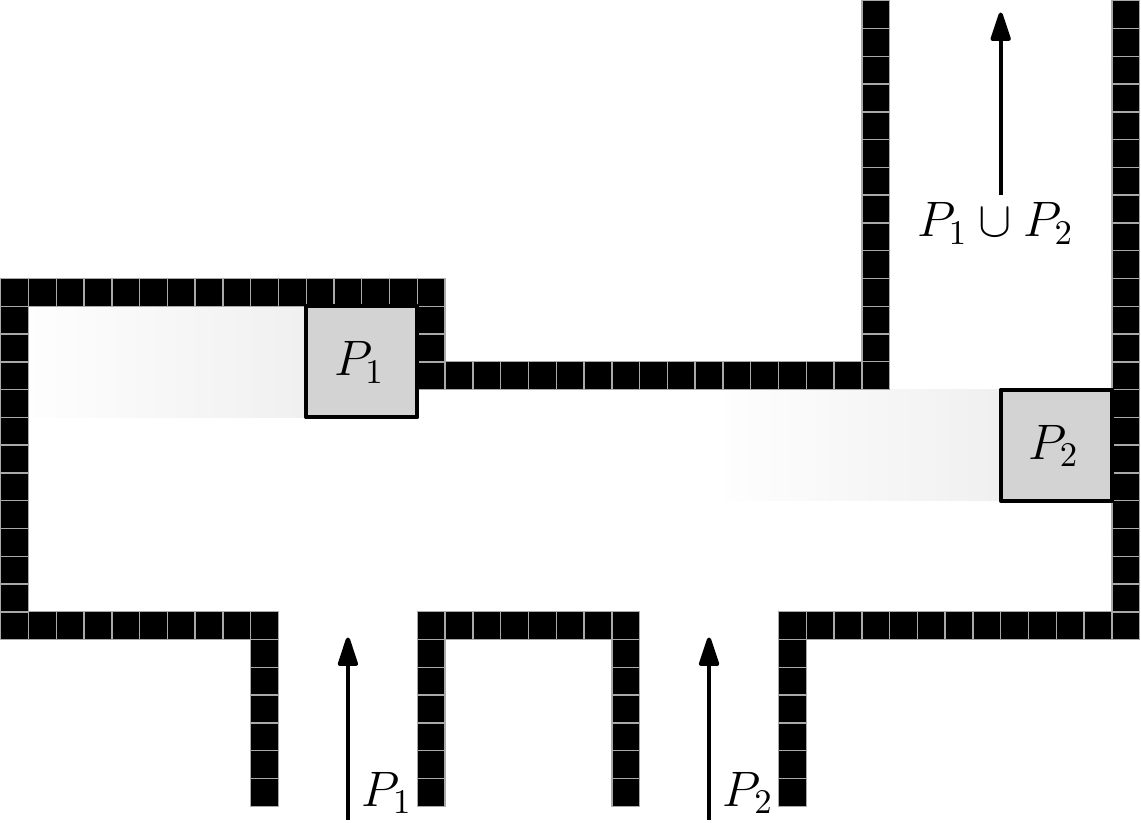}}
		\hfill
		\subfigure[Left]{
			\includegraphics[scale=0.235]{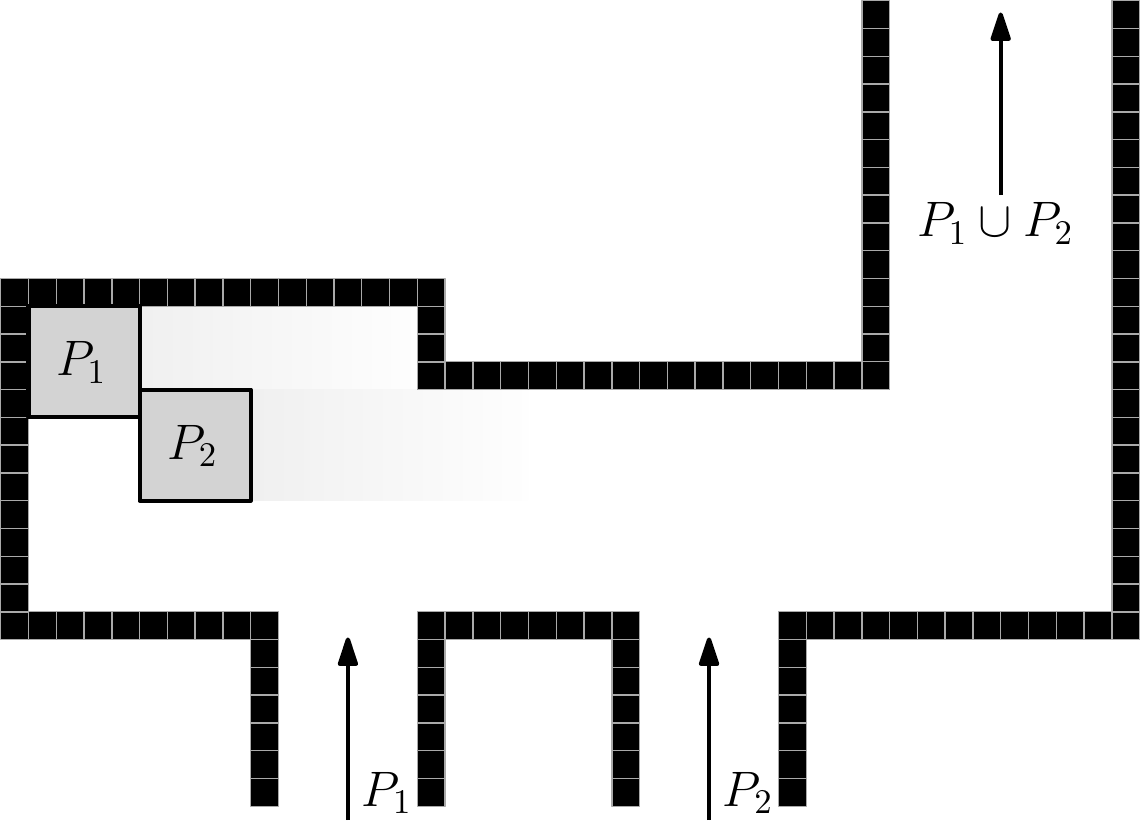}}
		\subfigure[Down]{
			\includegraphics[scale=0.235]{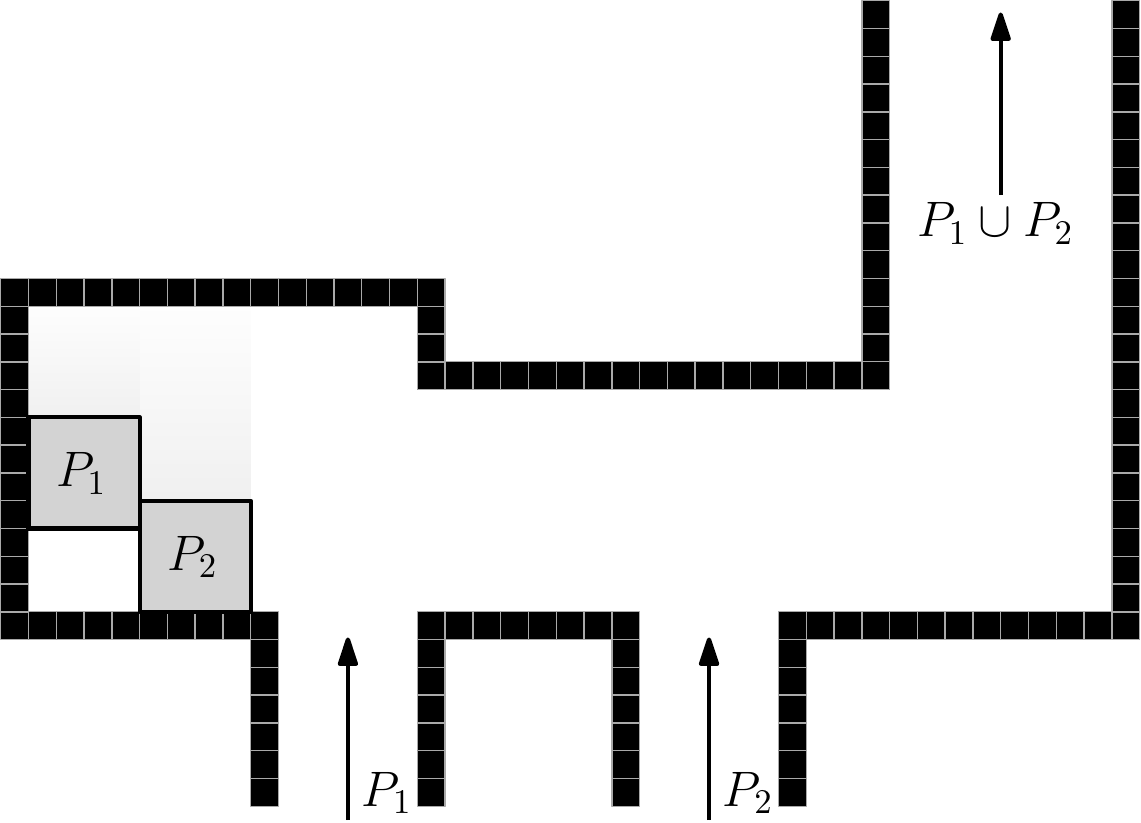}}
		\hfill
		\subfigure[Right]{
			\includegraphics[scale=0.235]{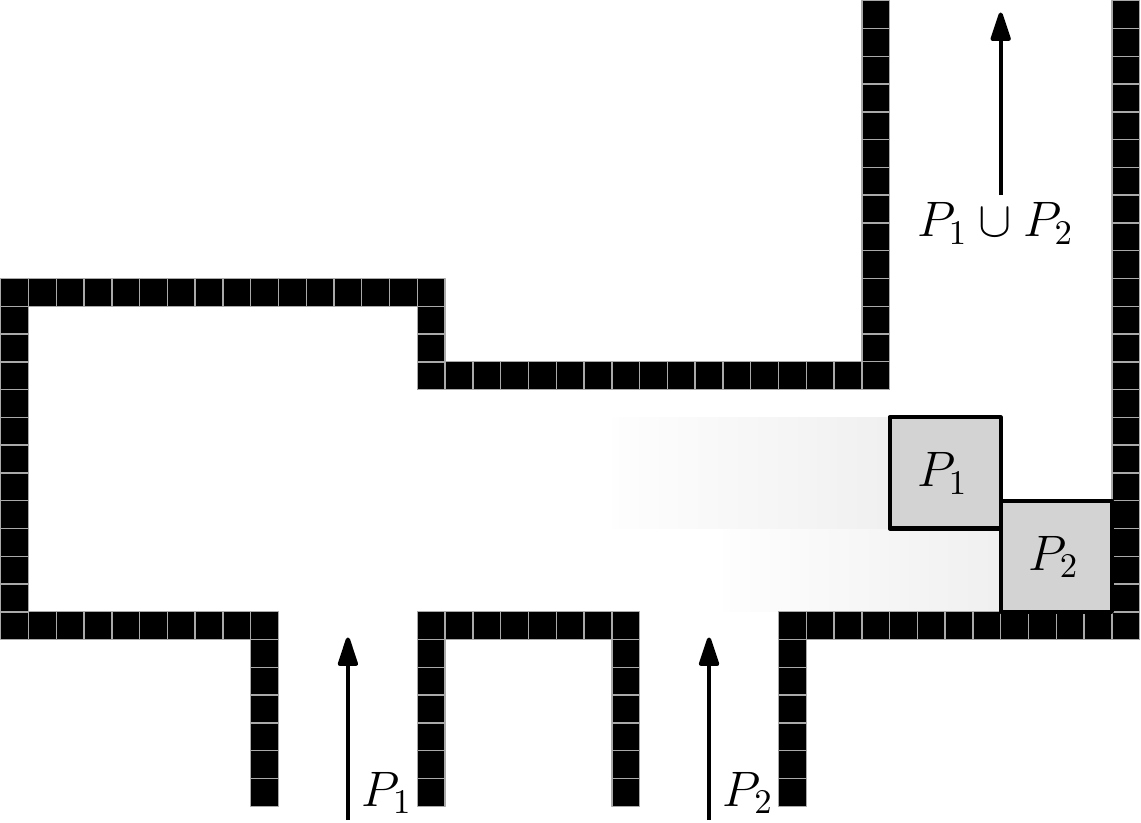}}
		\hfill
		\subfigure[Up]{
			\includegraphics[scale=0.24]{./figures/monotone_2_1.pdf}}
		\caption{Assembling two subpolyominoes $P_1$ and $P_2$, where the topmost tile of $P_1$ lies above the topmost tile of $P_2$. These are the same movements as seen in Fig.~\ref{fig:ortho_convex} \revision{for convex polyominoes. Thus, we can combine two subpolyominoes while constructing the next convex subpolyomino.}}
		\label{fig:monotone:p1p2}
		\vspace*{-4mm}
	\end{figure}
	
	\begin{figure}
		\centering
		\subfigure[Up]{
			\includegraphics[scale=0.26]{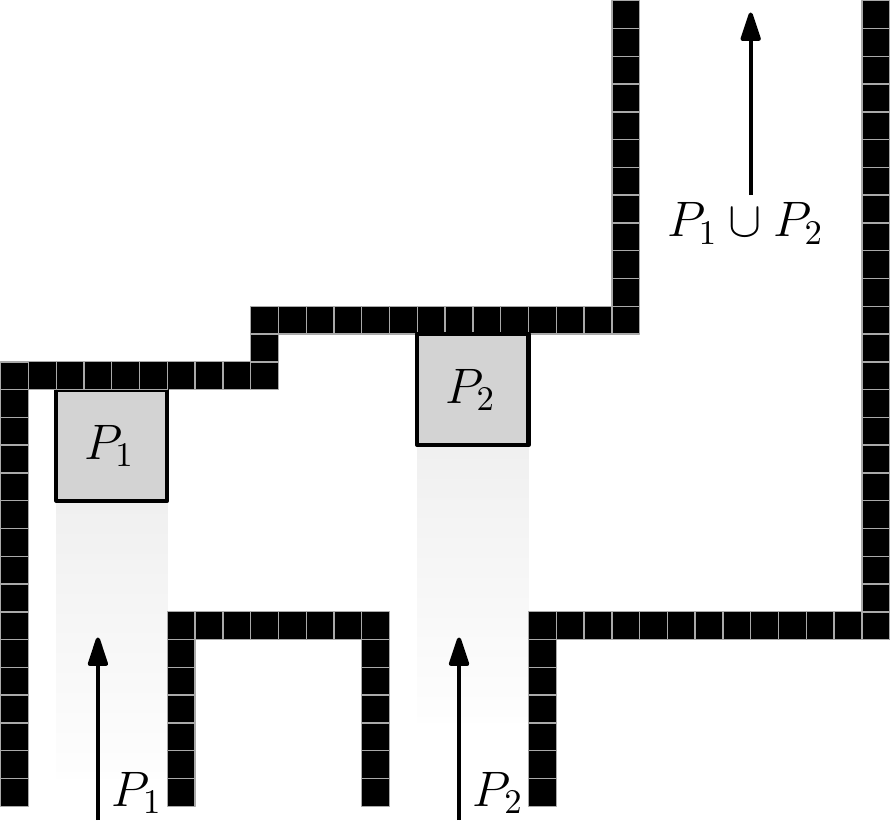}}\hfill
		\subfigure[Right]{
			\includegraphics[scale=0.26]{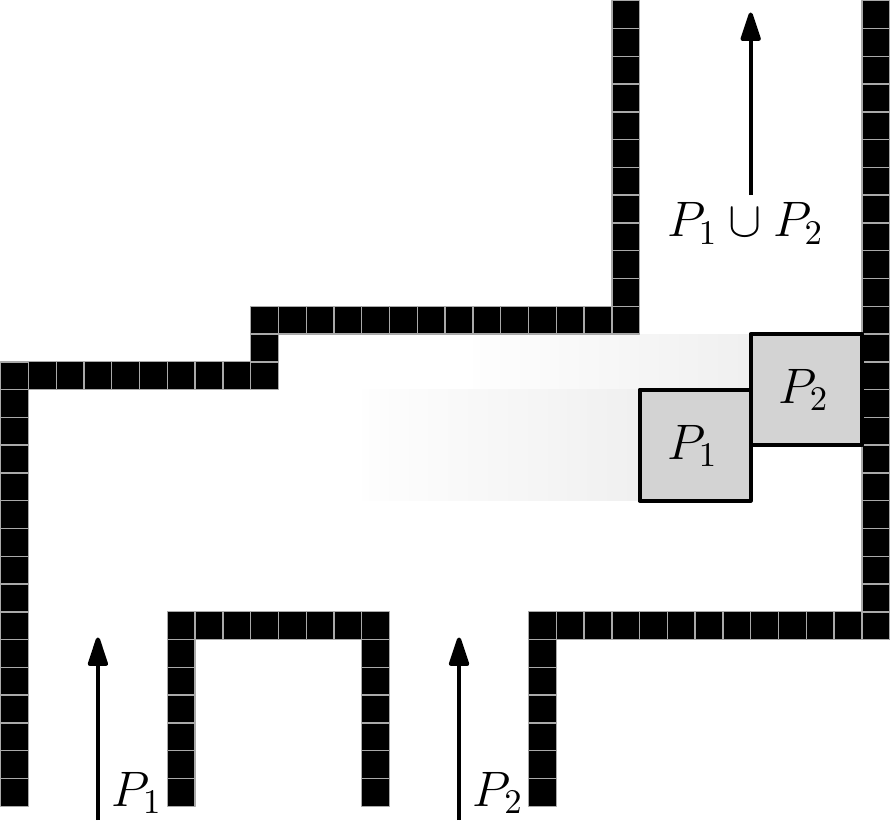}}	\hfill
		\subfigure[Left]{
			\includegraphics[scale=0.26]{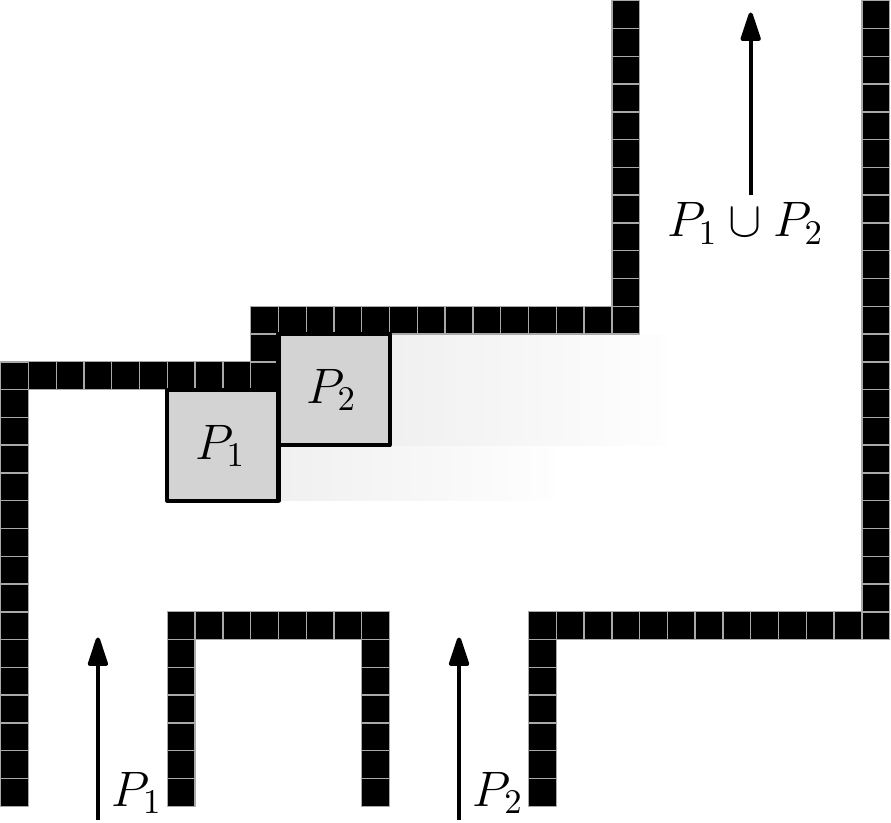}}	\hfill		
		\subfigure[Down]{
			\includegraphics[scale=0.26]{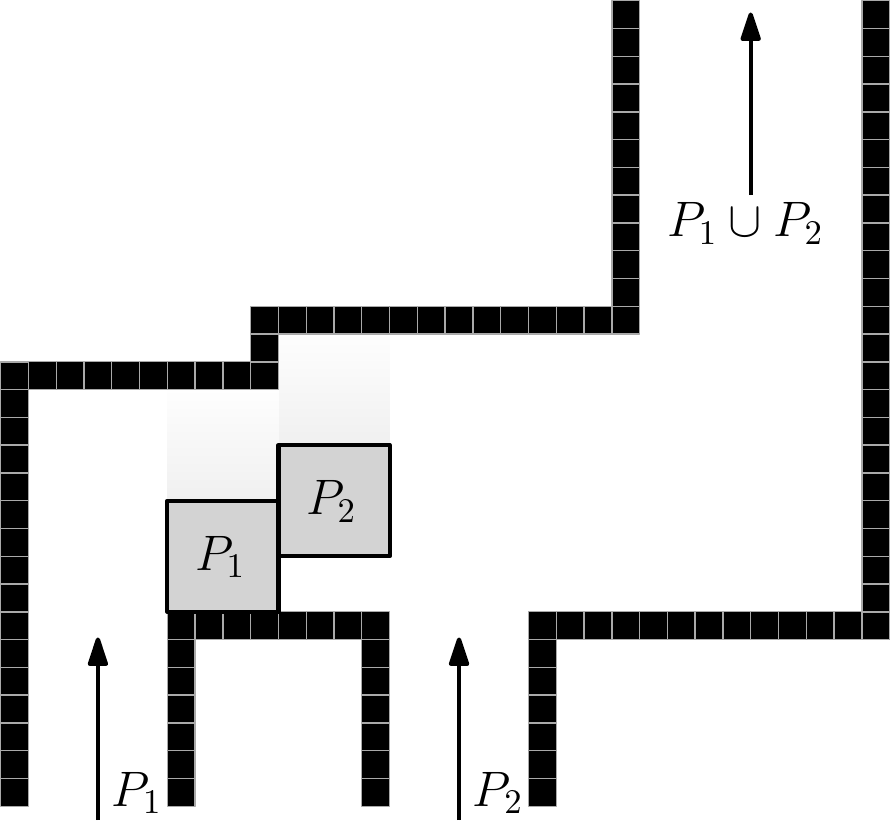}}\hfill
		\subfigure[Right]{
				\includegraphics[scale=0.26]{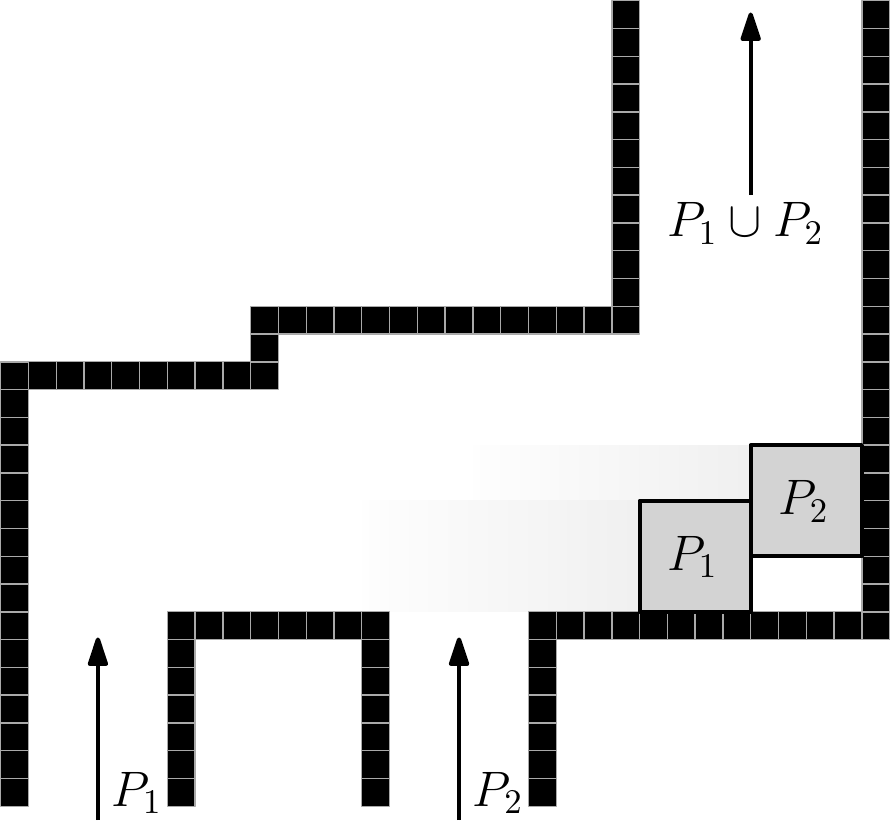}}\hfill	
		\subfigure[Up]{
				\includegraphics[scale=0.26]{./figures/monotone_1_1.pdf}}\hfill
		\caption{Assembling two subpolyominoes $P_1$ and $P_2$, where the topmost tile of $P_2$ lies above the topmost tile of $P_1$. These are the same movements as seen in Fig.~\ref{fig:ortho_convex} \revision{for convex polyominoes. Thus, we can combine two subpolyominoes while constructing the next convex subpolyomino.}}
		\label{fig:monotone:p2p1}
		\vspace*{-4mm}
	\end{figure}
		
	\begin{figure}
		\centering
		\includegraphics[angle=270, width=\columnwidth]{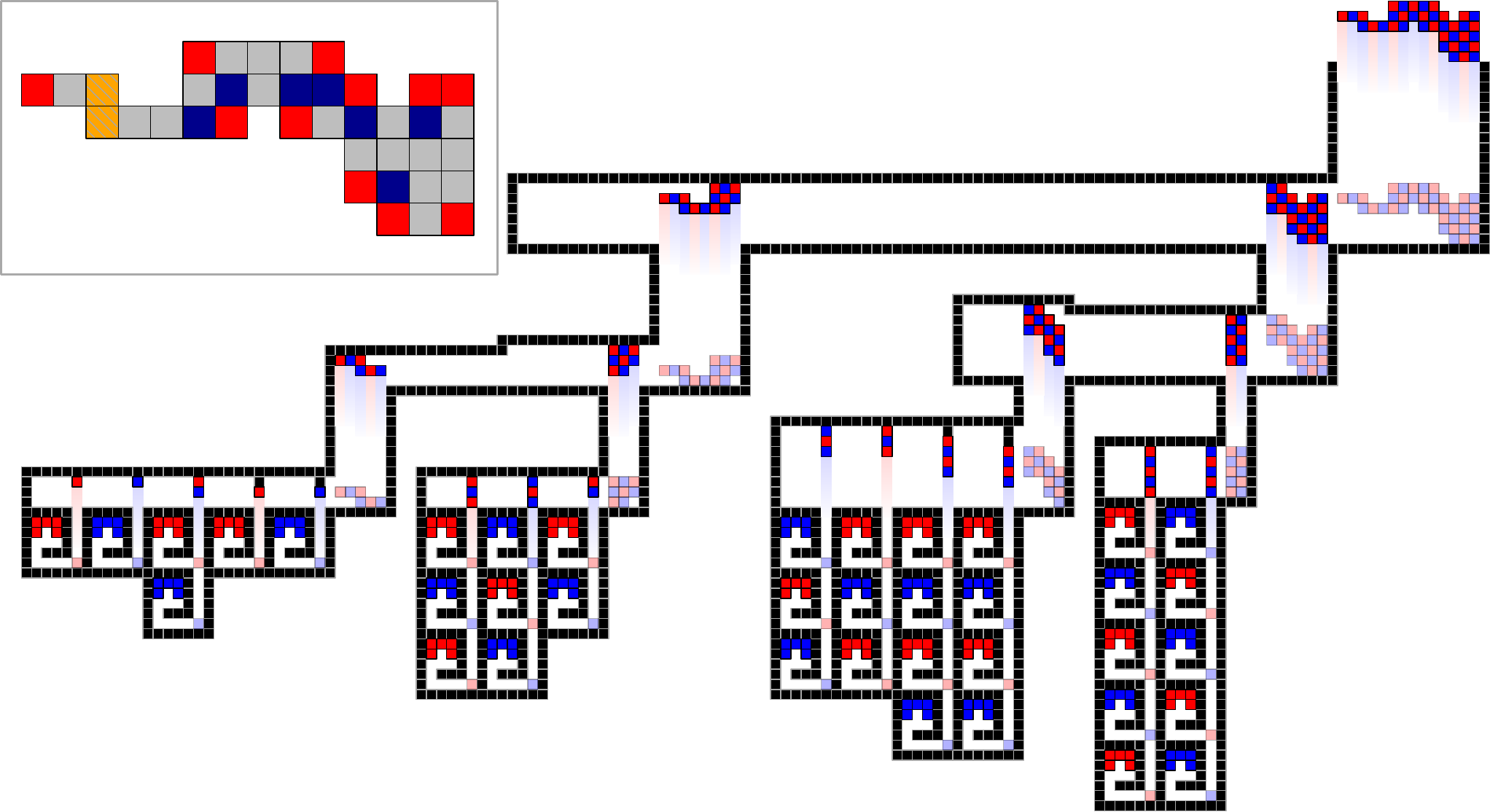}
		\caption{A complete example constructing $P$ with $\decomp(P)=3$. State shown is after an up-movement and its previous state in translucent colors. \revision{Top-left box: A polyomino $P$ with \revision{locally} convex tiles (red), \revision{locally} reflex tiles (blue), and tiles that are both \revision{locally} convex and \revision{locally} reflex (orange, striped).}}
		\label{fig:complete_monotone}
		\vspace*{-4.5mm}
	\end{figure}

%% file: 04-polyominoes.tex
	\section{Assembling Non-Monotone Shapes}
	In this section we show how to decide constructibility for special classes of polyominoes, namely simple polyominoes and polyominoes with convex holes.
	We end this section by showing how much space is needed for the workspace in which we can assemble the polyominoes.
	\vspace{-2mm}
	\subsection{Simple Polyominoes}
	To prove if a simple polyomino can be constructed we look at the converse process: a decomposition.
	As defined in the preliminary section we use 2-cuts to decompose a polyomino.
	If the polyomino cannot be decomposed by 2-cuts then the polyomino cannot be constructed by successively putting two subpolyominoes together.
	We show with the next lemma that we can greedily pick any valid straight 2-cut.
	
	\begin{lemma}
		Any valid straight 2-cut preserves decomposability.
	\end{lemma}
	
	\begin{proof}
		Consider a straight 2-cut $\ell$ and a sequence $\sigma=(\ell_1,\dots,\ell_m)$ of cuts, decomposing $P$ into single tiles.
		Assume $\ell$ is part of the cut sequence but not the first cut in $\sigma$.
		Then, there is a 2-cut $\ell'$ being made directly before $\ell$ in a polyomino $P^*$ induced by cuts before $\ell'$.
		We can now swap $\ell'$ and $\ell$ preserving their property of being 2-cuts: for $\ell$ we assume it is a 2-cut in $P$, which is also true in any subpolyomino induced by 2-cuts; 
		the same holds for $\ell'$, it is a 2-cut in $P^*$ and thus, also in any subpolyomino induced by 2-cuts. 
		After swapping both cuts we have the same decomposition yielding a valid decomposition of $P$.
		We can now repeat this procedure until $\ell$ is the first cut in $P$.
		
		However, $\ell$ may not be in the cut sequence $\sigma$.
		We now show that we can use $\ell$ as a cut by exchanging cuts.
		Let $\ell_k$ be the last cut intersecting $\ell$.
		This cut separates two cuts $\ell'$ and $\ell''$ which lie on $\ell$.
		Because $\ell$ is a 2-cut, also $\ell'\cup\ell''$ must be a 2-cut in the polyomino where we use cut $\ell_k$.
		Therefore, we can first use the cut $\ell'\cup \ell''$ and then the two cuts $\ell'_k$ and $\ell''_k$ induced by the intersection of $\ell$ and $\ell_k$.
		By repeating this procedure, we get $\ell$ as part of the cut sequence $\sigma$. 
	\end{proof}
		\begin{definition}
			A tile $t$ of a polyomino $P$ is said to be \emph{locally convex} if there exists a $2\times 2$ square solely containing $t$.
			If the square only contains $t$ and its two neighbors, then we call $t$ \emph{locally reflex}.
			Note that a tile can be \revision{locally} convex and \revision{locally} reflex at the same time (see Box in  Fig.~\ref{fig:complete_monotone}).
		\end{definition}
	
	\begin{lemma}\label{lem:reflex_only}
		%It is sufficient to consider straight 2-cuts cutting at a reflex tile.
		\revision{Any non-convex, straight 2-cuttable polyomino $P$ can be decomposed into convex subpolyominoes by only using straight 2-cuts cutting along a \revision{locally} reflex tile.}
	\end{lemma}
	
	\begin{proof}
		W.l.o.g., consider a vertical straight 2-cut $\ell$ that may not cut along a \revision{locally} reflex tile.
		Then we can move a cut $\ell'$ to the left starting at $\ell$ until we reach a \revision{locally} reflex tile \revision{$t$ such that the cut goes through the corner of $t$ that lies on the boundary of $P$} (if we cannot reach a \revision{locally} reflex tile we move $\ell'$ to the right).
		We obtain three subpolyominoes: $P_1$ to the right of $\ell$, $P_2$ to the left of $\ell'$, and $P_3$ between $\ell'$ and $\ell$ (see Fig.~\ref{fig:2-cut:reflex}).
		
		Assume $\ell'$ is not a valid 2-cut, i.e., a tile is blocked in $P_2$ or $P_3$. (If there is a blocked tile in $P_1$, then also $\ell$ would not be a 2-cut.)
		Consider the first case, where $P_2$ has a blocked tile $t$.
		Then, $t$ has an $y$-coordinate which is at most as high as the highest tile in $P_3$ plus 1 and at least as high as the lowest tile in $P_3$ minus one (or else both blocking tiles must be in $P_1$ and thus, $\ell$ would be no 2-cut).
		Let $q_1\in P_3$ to the right of $t$ and $q_2\in P_1\cup P_3$ to the left of $t$ be the two tiles blocking $t$. 
		By replacing $q_1$ with its right neighbor we still have two tiles blocking $t$.
		Because $\ell$ is a 2-cut we can repeat this procedure  until $q_1\in P_1$.
		We can repeat the procedure for $q_2$ if $q_2\in P_3$.
		Thus, both blocking tiles are in $P_1$ and $\ell$ cannot be a 2-cut.
		
		For the second case the blocked tile $t$ lies in $P_3$.
		Then, also the right neighbor $t'$ of $t$ is blocked. 
		This is also true for $t'$.
		Therefore, we can go to the right until we reach $P_1$ and thus, there is a tile in $P_1$ which is blocked.		
		This means, also $\ell$ cannot be a valid 2-cut, which is a contradiction to $\ell$ being a valid 2-cut.
		
		\revision{As each cut $\ell'$ reduces the number of \revision{locally} reflex tiles by at least one, the remaining polyominoes will be convex after a limited number of cuts.}
	\end{proof}
	
	\begin{figure}
		\begin{varwidth}[t]{0.4\columnwidth}
		\includegraphics[width=\columnwidth]{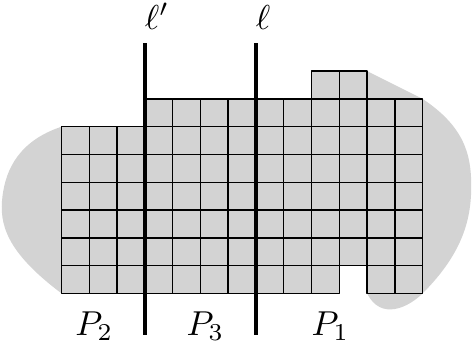}
		\caption{\revision{The original cut $\ell$ and its shifted copy $\ell'$, which together  split the polyomino into three parts $P_1,P_2,P_3$.}}
		\label{fig:2-cut:reflex}
		\end{varwidth}
		\hfill
		\begin{varwidth}[t]{0.5\columnwidth}
				\includegraphics[width=0.9\columnwidth]{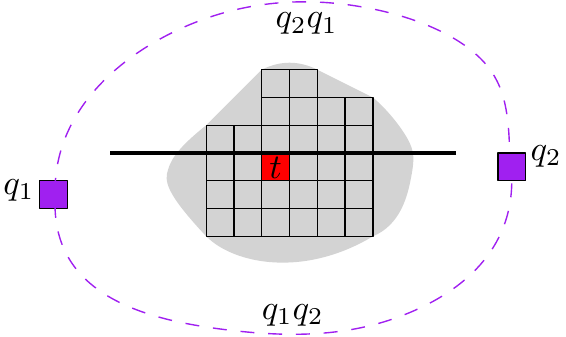}
				\caption{A not \revision{locally} convex tile $t$ (red) in $P_1$ (gray area) blocked by $q_1$ and $q_2$ (purple).
					If the path $q_2q_1$  exists, there is at least one blocked \revision{locally} convex tile above the black bold line. \revision{If $q_1q_2$ exists, we proceed analogously.}}
				\label{fig:2-cut:convex_tiles}
		\end{varwidth}
		\vspace*{-5mm}
		\end{figure}

	\begin{lemma}\label{lem:convex}
		It is sufficient to consider \revision{locally} convex tiles for checking if a cut $\ell$ is a valid straight 2-cut.
	\end{lemma}
	
	\begin{proof}
		Assume w.l.o.g. $\ell$ is a vertical cut splitting the polyomino in two subpolyominoes $P_1$ and $P_2$.
		W.l.o.g., consider a not \revision{locally} convex tile $t\in P_1$ blocked by two tiles $q_1,q_2 \in P_2$. 
		Because $\ell$ is a 2-cut and $P$ is simple, there must be a path from $q_1$ to $q_2$ within $P_2$.
		This path must go around $P_1$ either above or beneath $t$ (see Fig.~\ref{fig:2-cut:convex_tiles}).
		
		In case the path moves above $t$, consider a horizontal cut directly above $t$ (see Fig.~\ref{fig:2-cut:convex_tiles}).\todo{Figure 8 caption: The last sentence seems incomplete.}
		This cut splits $P_1$ into components.
		In each component there are at least four \revision{locally} convex tiles from which at most two became \revision{locally} convex through the cut.
		Thus, two of these \revision{locally} convex tiles were also \revision{locally} convex in $P_1$.
		It is easy to see in the figure that both \revision{locally} convex tiles are also blocked by tiles on the path from $q_1$ to~$q_2$.
		
		In the second case we proceed analogously with the difference that we use a horizontal cut directly below $t$.
		We conclude that in any case there is a \revision{locally} convex tile in $P_1$ that is being blocked if there is a blocked, not \revision{locally} convex tile.
		Note that the other direction may not be true.
	\end{proof}
	
\rem{	\begin{figure}
		\centering
		\includegraphics[width=0.5\columnwidth]{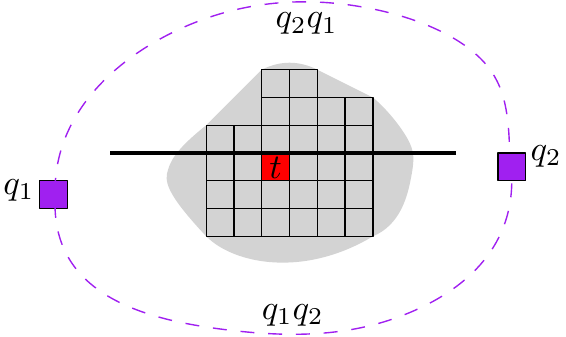}
		\caption{A not \revision{locally} convex tile $t$ (red) in $P_1$ (gray areay) blocked by $q_1$ and $q_2$ (purple).
			If there is a path $q_2q_1$ then there is at least one \revision{locally} convex tile above the black bold line that is also blocked. Analogously, if $q_1q_2$ exists.}
		\label{fig:2-cut:convex_tiles}
		\vspace*{-6mm}
	\end{figure}
	}

	\begin{lemma}\label{lem:sweep}
		Checking if a 2-cut $\ell$ is valid can be done in $O(N+r\log r)$ time, where $r$ is the number of \revision{locally} reflex tiles.
	\end{lemma}
	
	\begin{proof}
		W.l.o.g. assume $\ell$ to be a vertical straight cut and also assume that we are checking blue tiles only.
		As a first step we scan through the polyomino and search for all tiles that represent a corner, i.e., the tile is \revision{locally} convex or \revision{locally} reflex.
		Additionally, we can store the neighbor corner tiles of each corner tile (these are up to four tiles).
		Both steps can be done with one scan, and thus in $O(N)$ time.
		
		Now, consider the cut $\ell$ splitting the polyomino into subpolyominoes $P_1$ and $P_2$.
		Finding the corner tiles in $P_1$ and $P_2$ can be done in $O(r)$ time by a breadth-first search.
		We proceed with the following procedure for $P_1$ (analogously for $P_2$):
		\begin{enumerate}
			\item\label{enum:sweep:vert} Get all vertical lines connecting two corner tiles in $P_2$ and stretch this line by one tile if a corner tile is red (this checks if a blue tile would pass a red tile).
			\item Sort the set $C_r$ of corner tiles in $P_2$ lexicographically by $y$-coordinate and then by $x$-coordinate.
			\item Start a sweep line from bottom to top having the tiles in $C_r$ as event points.
			\item \label{enum:sweep:sweep} On each event point $p$ do the following update:
			\begin{itemize}
				\item If $p$ is a start point of a vertical line but lies left of the current vertical line, remove $p$ from $C_r$.
				\item If $p$ is a start point and lies to the right of the current line add the tile of the current line to $C_r$ and jump to the new vertical line.
				
				\item If $p$ is an end point of the current vertical line, then jump to the nearest vertical line to the left and add the tile of this line to $C_r$.
				\item If $p$ is an end point but not of the current vertical line, remove $p$ from $C_r$.
			\end{itemize} 
			\item Repeat steps 1--4, switching left and right, to get $C_l$ 
			\item \label{enum:sweep:check} For each \revision{locally} convex tile $t$ in $P_1$:
			\begin{itemize}
				\item find $q_1\in C_r$ having highest $y$-coordinate below $t$, and $q_2\in C_r$ having lowest $y$-coordinate above $t$. (Both shall be the left-most tile in case of ties.)
				\item find $q'_1\in C_l$ having highest $y$-coordinate below $t$, and $q'_2\in C_l$ having lowest $y$-coordinate above $t$. (Both shall be the left-most tile in case of ties.)
				\item If $t$ lies to the left of segment $q_1q_2$ and to the right of segment $q'_1q'_2$ return false.
			\end{itemize} 
		\end{enumerate}
		
		\begin{figure*}
			\centering
			\subfigure[$t_1$ is blocked to the right because it lies to the left of the red line. $t_2$ and $t_3$ are on the right side and there is no segment blocking them to the right. The leftmost segment belongs to bounding box of $P_2$.]{\hspace{1cm}
				\includegraphics[width=0.6\columnwidth]{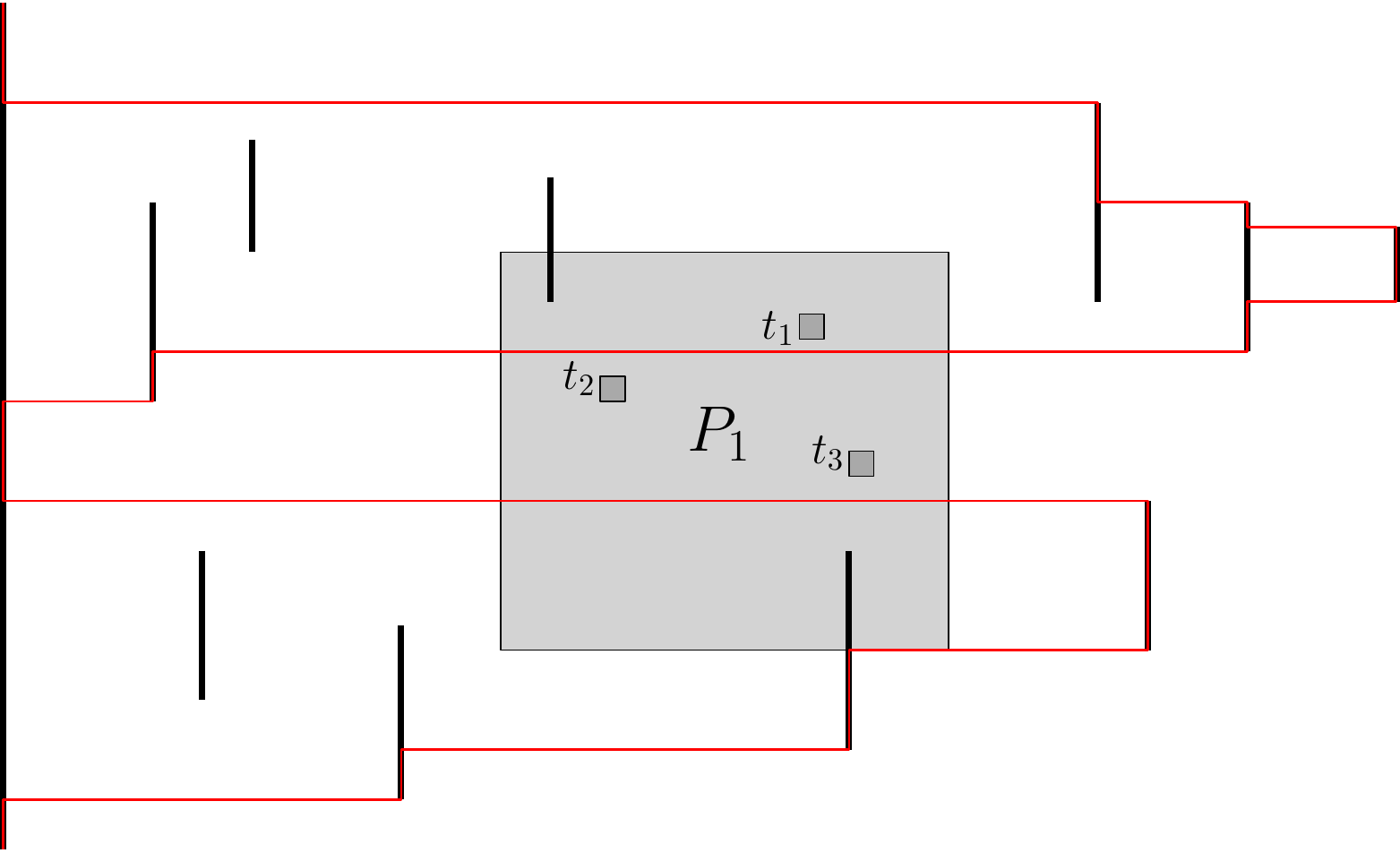}\hspace{1cm}}
			\hfil
			\subfigure[$t_1$ and $t_2$ are blocked to the left because they lie to the right of the red line. Only $t_3$ is on the left side and there is no segment blocking it to the left. The rightmost segment belongs to bounding box of $P_2$.]{\hspace{1cm}
				\includegraphics[width=0.6\columnwidth]{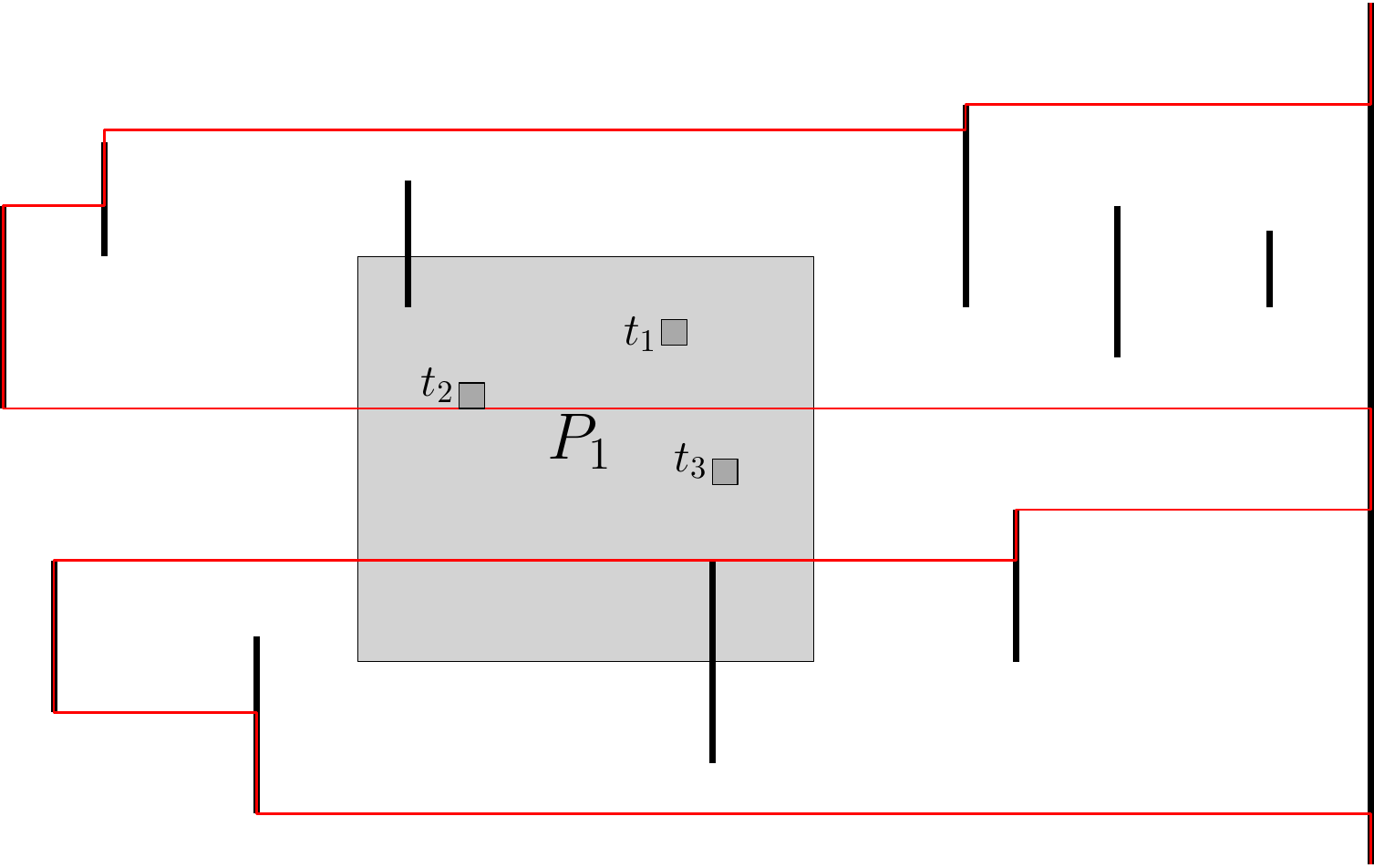}\hspace{1cm}}
			\caption{Example for the data structure used in Lemma~\ref{lem:sweep}. We observe that $t_1$ is always on the wrong side of the red line and is thus blocked in both directions. Vertical lines are part of $P_2$.}
			\label{fig:datastructure}
			\vspace{-0.4cm}
		\end{figure*}
		
		This computes a left and right envelope of vertical lines in $P_1$ and $P_2$, respectively. 
		This allows an easy check if there is a tile on the left/right blocking a tile from $P_1$ in this direction (for an example, see Fig.~\ref{fig:datastructure}). 
		
		The runtime is in $O(r\log r)$:
		Step~\ref{enum:sweep:vert} needs $O(r)$ time because there are $O(r)$ corner tiles and at most two vertical lines per corner tile.
		Sorting a set lexicographically in two dimensions can be done in $O(r\log r)$.
		With a careful view on step~\ref{enum:sweep:sweep}, we can observe that each update of the $O(r)$ event points costs $O(\log r)$ and thus in total $O(r\log r)$ time.
		Step~\ref{enum:sweep:check} can be done in $O(\log r)$ time for each \revision{locally} convex tile.
		Therefore, we need $O(r\log r)$ time in total.
	\end{proof}
	
	\revision{The next theorem is straightforward to prove.} 
	
	\begin{theorem}\label{th:simple:straight}
		Let $r$ be the number of \revision{locally} reflex tiles. We can find a valid straight 2-cut in $O(N+r^2\log r)$ time.
	\end{theorem}
	
	\rem{
	\begin{proof}
		For any straight cut $\ell$ check if $\ell$ is a valid 2-cut.
		There are $O(r)$ many cuts with $O(N+r\log r)$ time to check if it is a 2-cut.
		We can do the initial scan over the polyomino from Lemma~\ref{lem:sweep} before the first check and reduce the time to check for a 2-cut to $O(r\log r)$.
		In total we have a runtime of $O(N+r^2\log r)$.
	\end{proof}
}
	
	\begin{theorem}
		A decomposition tree of valid 2-cuts for a polyomino $P$ can be used to build a labyrinth constructing $P$. This labyrinth can also be used for pipelining.
	\end{theorem}
	
	\begin{proof}
		Consider a cycle of the seven unit steps \textit{right, up, down, up, right, left, down}.
		This is the movement sequence which was already seen for convex and monotone polyominoes but with two more movements.
		This cycle preserves the ability to construct monotone polyominoes in the labyrinth above. 
		Also observe that turning the gadgets seen in Figs.~\ref{fig:monotone:p1p2}~and~\ref{fig:monotone:p2p1} by 90 degrees clockwise yields gadgets that put two polyominoes on top of each other.
		
		Transforming a decomposition tree of 2-cuts for a polyomino $P$ can easily be done:
		Consider the layers of the decomposition tree, with the root being layer zero, its children being layer one, and so on.
		In each vertex in one layer either a horizontal or vertical cut is made.
		Corresponding to this cut we construct a gadget putting the two children of this vertex together.
		At some point only monotone subpolyominoes exist.
		These can be build using the methods described above.
		
		The length of a root-leaf-path may vary.
		In this case we can build loops so we can put two polyominoes together at the right time.
	\end{proof}
	
	\begin{theorem}
		Any straight 2-cuttable polyomino $P$ can be build within $O(r)$ unit steps, where $r$ is the number of \revision{locally} reflex tiles in $P$.  $D$ copies require $O(r+D)$ unit steps.
	\end{theorem}
	
	\begin{proof}
		\revision{
		Doing cuts along locally reflex tiles reduces the number of locally reflex tiles by at least one. This implies a maximum depth of $O(r)$ of the decomposition tree and thus, $O(r)$ cycles to produce $P$.
		As seen before, pipelining yields a construction time of $(r+D)$ unit steps, which is an amortized constant construction time if $D\in\Omega(N)$.} 
		\rem{
		 We choose cuts only at \revision{locally} reflex tiles, so each cut reduces the number of \revision{locally} reflex tiles $r$ by at least one, with 
		no more than $r$ cuts.
		When these cuts are made, only convex polyominoes are left which can be assembled in $O(1)$ unit steps.
		However, we may only be able to combine two subpolyominoes after one cycle. 
		This results in $O(r)$ cycles.
		Pipelining yields, as seen before, a construction time of $O(r+D)$ unit steps, which is an amortized constant number of unit steps per copy if $D\in \Omega(r)$.
	}
	\end{proof}
	
	Unfortunately, the number of \revision{locally} reflex tiles $r$ can be in $\Omega(N)$ and thus, we may need $\Omega(N)$ cuts to build the polyomino.
	In particular, Fig.~\ref{fig:2-cut:worstcase} left shows an example which needs $\Omega(N)$ cycles to build. Even scaling by some factor $k$, i.e., replacing each tile by an $k\times k$ supertile, seems not to help.
	Moreover, there are also polyominoes we cannot build by putting two subpolyominoes together at the same time (see Fig.~\ref{fig:2-cut:worstcase} right).
	
	%There is a 2-cuttable polyomino that is constructible in $\Omega(N)$ unit steps, if we are allowed to use 2-cuts only.
	%Consider the polyomino $P$ in Fig.~\ref{fig:2-cut:worstcase}.
	%	Because we are using 2-cuts only, we are only able to cut off one arm after another from right to left.
	%	Due to its constant height and the constant width, we can cut off a constant number of tiles at a time.
	%	Thus, we get a decomposition tree of depth $\Omega(N)$ and we need $\Omega(N)$ unit steps to construct $P$.

	\begin{figure}
		\centering
		\includegraphics[width=0.5\columnwidth]{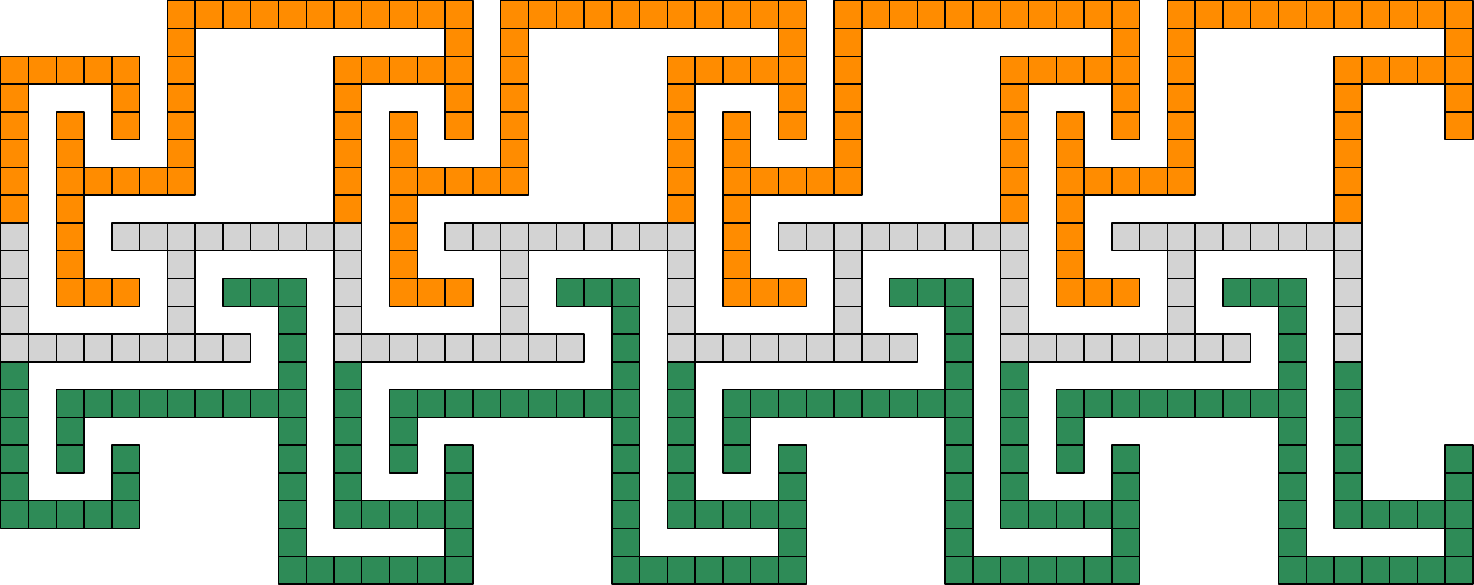}
		\hfill
		\includegraphics[width=0.25\columnwidth]{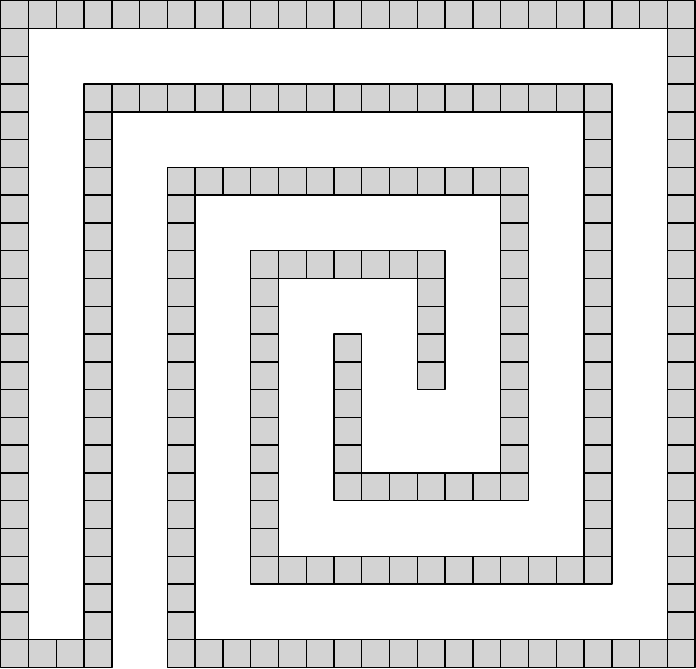}
		\caption{Left: A polyomino needing $\Omega(N)$ steps to build as we cannot separate the green nor the orange part efficiently from the grey part. 
			Right: Polyomino which is not 2-cuttable. Any cut splits the polyomino either in two subpolyominoes which cannot be pulled apart or into more than two subpolyominoes.}
		\label{fig:2-cut:worstcase}
		\vspace*{-6mm}
	\end{figure}
	
	\subsection{Non-Straight Cuts}
	
	Considering any 2-cut makes it more difficult to find cuts, as there are exponential many possible cuts. 
	However, we do not need to consider all cuts.
	For a given start $s$ and end $e$ on the boundary of a polyomino $P$, we can show that it is sufficient to consider only one cut connecting $s$ and $e$.
		The proof is similar to the one of Lemma~\ref{lem:reflex_only}.
	
	\rem{
	\begin{lemma}
		Given two points $s$, $e$ on the boundary of a polyomino $P$ and a cut $\ell$ connecting $s$ and $e$, the cut $\ell$ is a valid 2-cut if and only if any other 2-cut connecting $s$ and $e$ is valid. 
	\end{lemma}
	}
	
%	\begin{proof}
		%This proof is quite similar to the proof of Lemma~\ref{lem:reflex_only}.
		%W.l.o.g. let $\ell$ be a vertical cut. 
		%Let $\ell'$ be a second cut connecting $s$ and $e$ different to $\ell$. 
		%Define three areas $P_1$, $P_2$ and $P_3$ as follows:
		%$P_1$ contains any tile that lies on the right side of both cuts,
		%$P_2$ contains any tile that lies to the left of both cuts, and $P_3$ contains any tile that lies between both cuts.
		%As we are considering simple polyominoes only, any horizontal line between $\ell$ and $\ell'$ contains no empty position.
		%With the argumentation for Lemma~\ref{lem:reflex_only} the lemma follows.
	%\end{proof}
	
	\begin{theorem}\label{th:simple:non-straight}
		Given a 2-cuttable polyomino $P$, we can find a 2-cut in time $O(r^2N\log N)$, where $r$ is the number of \revision{locally} reflex tiles in $P$.
	\end{theorem}
	
	\begin{proof}
		The idea of this proof is to find $O(r^2)$ 2-cuts which are then tested if they are valid.
		\revision{One necessary criterion is that no cut moves three units to the left or right in case of vertical cuts.
			This can be achieved with a directed graph $D_P$. %(due to space constraints we omit details).
			As seen in Fig.~\ref{fig:cut_graph}, we add a set of $O(r)$ vertices that correspond to corners 
of tiles lying on the boundary of $P$ (giving rise to the set $V_B$), or that correspond to corner tiles 
not lying on the boundary (resulting in the set $V_I$).
			We add edges between \emph{adjacent} vertices with weight $\frac 1{2N}$ if both vertices are in $V_I$. If both vertices are in $V_B$, then the edge has weight $2$, otherwise~$1$. 
			
			A 2-cut is represented by a shortest path of weight at most $2.5$ containing exactly two vertices of $V_B$. 
			If we have at least three vertices of $V_B$ in the shortest path, it has length at least 3.
			Thus, paths from one vertex in $V_B$ to another vertex of $V_B$ define cuts going through $P$.
			Finding all shortest paths from one vertex in $V_B$ lasts $O(N\log N)$ time, as there are $O(N)$ edges in $D_P$. This implies a total time of $O(rN\log N)$ for finding all shortest paths of length at most $2.5$.
			
			Because one cut can make $O(N)$ turns, checking whether the cut is valid takes time $O(N\log N)$. Thus, checking all $O(r^2)$ cuts if they are valid needs time $O(r^2N\log N)$. 
			}	
		\begin{figure}
			\centering
			\includegraphics[width=0.39\columnwidth]{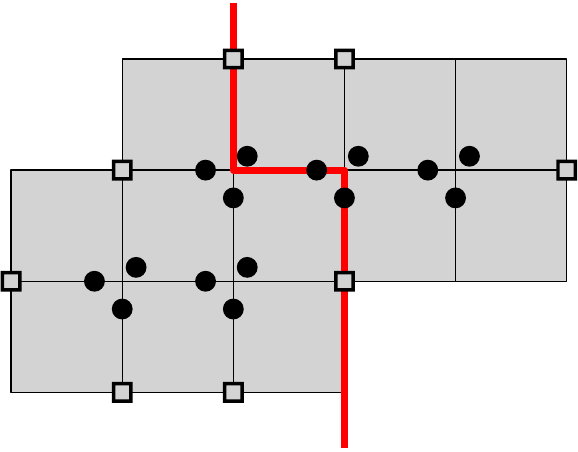}\hfill
			\includegraphics[width = 0.55\columnwidth]{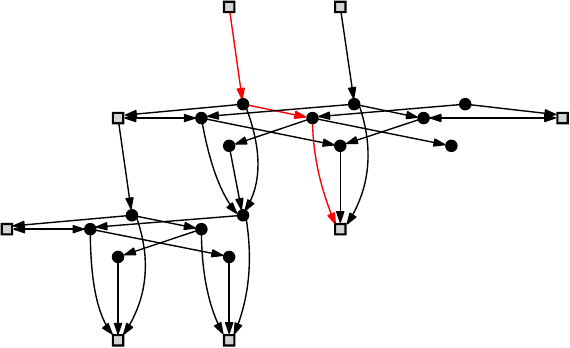}
			\caption{A polyomino $P$ (grey tiles) and the graph $D_P$ (right). The vertices in $V_B$ are shown as squares, the vertices in $V_I$ are shown as disks. Red bold line in $P$ is a 2-cut. Red path in the graph represents this cut.}
			\label{fig:cut_graph}
		\vspace*{-5mm}
		\end{figure}	
	%%% Restart Remove
	\rem{	
		Having this graph, we can now go on finding paths from each vertex $v\in V_B$ to another vertex $w\in V_B$ defining a cut $\ell_{v,w}$ in $P$. 
		This can easily be done with Dijkstra's Algorithm for finding shortest paths.
		We start the algorithm from each vertex $v\in V_B$ and store each shortest path to another $w\in V_B$ with weight at most 2.5 or else we would have no 2-cut.
		Because there are $O(N)$ many edges and vertices in total, we need $O(rN\log N)$ time to find all desired paths.
		
		It is left to check if there is a path which corresponds to a valid 2-cut.
		Consider some path from $v$ to $w$ and the cut $\ell_{v,w}$.
		Let $P_1$ and $P_2$ be the two subpolyominoes induced by $\ell_{v,w}$.
		Because $\ell_{v,w}$ may has $O(N)$ turns, $P_1$ and $P_2$ can have $O(N)$ corner tiles.
		Therefore a check with the algorithm used in Lemma~\ref{lem:sweep} needs $O(N\log N)$ time.
		In total, as there are $O(r^2)$ many paths, we need $O(r^2N\log N)$ time.
	}
	%%%End Remove
	\end{proof}
	
	\revision{All techniques can be generalized for polyominoes with convex holes.
	However, this increases the number of possible cuts to be checked.}
	In particular, there can be $O(r_h^3)$ possible ways to go through a hole $h$ with $r_h$ \revision{locally} reflex tiles. 
	Thus, the time to find a cut takes $O(N+r^3\log r)$ using straight cuts and $O(r^3N\log N)$ using non-straight cuts.

	\subsection{Workspace Size and Number of Obstacles}
%	\todo[inline][inline]{Intro text for this section?}
		\begin{theorem}
			Let $P$ be a polyomino. 
			Then, the workspace needed to assemble $D$ copies of $P$ can be put into a rectangle of width
			$O(w_P\mathcal{L}_P\cdot(\mathcal{C}_P + \sqrt D))$ and height
			$O(h_P\cdot(\mathcal{C}_P + \sqrt D))$, 
			where $w_P$ and $h_P$ are the width and height of $P$, $\mathcal{C}_P$ is the number of movement steps needed, and 
			$\mathcal{L}_P$ is the number of cuts made to decompose $P$ into convex subpolyominoes. 
			Furthermore, we only need $O(N(\mathcal{L}_P+\sqrt{D}))$ obstacles in the workspace.
		\end{theorem}
		
		\begin{proof}
%			The main idea for this proof is to look at an abstraction of the workspace:
			Represent each gadget as a block.
			An example block diagram shown in Fig.~\ref{fig:monotone:block} illustrates the structure of the workspace with width and height of each stage.
			
			Consider the decomposition tree $T$ of $P$ induced by cuts whose leafs are convex polyominoes.
			For convex polyominoes we can use the construction from Lemma~\ref{lem:convex} and for each inner node of $T$ we use the gadgets used in Theorem~\ref{th:monotone} to combine two subpolyominoes.
			Let $P_1,\dots,P_k$ be the convex polyominoes in the leafs of $T$ with width $w_1,\dots,w_k$ and height $h_1,\dots,h_k$.
			To construct one $P_i$, we need $O(\sqrt D w_i)\times O(\sqrt D h^*_0)$ space, where $h^*_0$ is the maximum height of all $P_i$.

			Now consider the $j$-th stage with $j\leq \mathcal{C}_P$ where some polyominoes are combined.
			Let $P'_1$ and $P'_2$ be two such polyominoes.
			After assembling these polyominoes the width of the polyomino $P''_1$ increases to $w''_1 \leq w'_1+w'_2$.
			Thus, the width of the workspace increases by at most $w'_1+w'_2$.
			We observe that any width of $P_1,\dots,P_k$ appears at most $\mathcal{C}_P+1$ times.
			With $w_i \leq w_P$ and $k\in O(\mathcal{L}_P)$, 
			this results in a total width of
			$\sum_{i=1}^{\mathcal{L}_P}w_P(\sqrt{D}+\mathcal{C}_P+1) \in O(w_P\mathcal{L}_P(\sqrt D + \mathcal{C}_P))$.
			
			For the total height consider the maximum height $h^*_j$ of all polyominoes in stage $j\leq \mathcal{C}_P$. 
			Because we need $O(h^*_j)$ space in the vertical direction for stage $j$, we have have a total height of $h_P\sqrt{D}+\sum_{j=1}^{\mathcal{C}_P}h^*_j \in O(h_P(\sqrt{D}+\mathcal{C}_P))$ resulting in a rectangle of size $O(w_P(\sqrt D + \mathcal{C}_P))\times O(h_P(\sqrt D + \mathcal{C}_P))$ enclosing the workspace.
			
			Although the workspace may be large, the number of obstacles needed is smaller.
			First, ignore any obstacle not needed as a stopper (see Fig.~\ref{fig:gadget_cropped}).
			This reduces the number of obstacles to $O(w_P+h_P)$. 
			Because $w_P,h_P\leq N$ this is $O(N)$.
			The same can be done for building the convex polyominoes.
			However, to keep the $D$ tiles in a container we need all $O(\sqrt D)$ obstacles.
			Thus, we have $O(N\sqrt D)$ obstacles to build all convex polyominoes and $O(\mathcal{L}_PN)$ obstacles for the gadgets which is in total $O(N(\mathcal{L}_P+\sqrt{D}))$.
		\end{proof}
		
		\begin{figure}
			\centering
			\includegraphics[width=0.7\columnwidth]{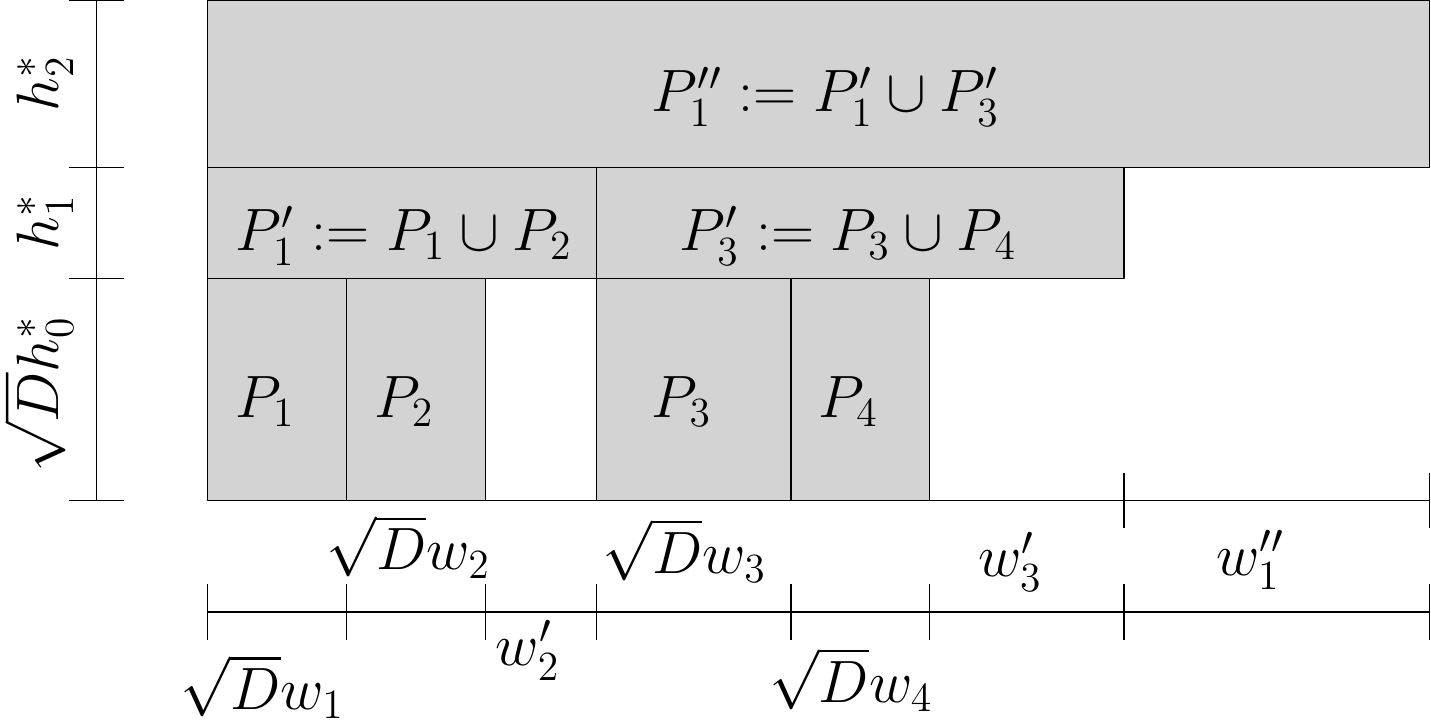}
			\caption{Block diagram of the workspace to construct a monotone polyomino.}
			\label{fig:monotone:block}
		\vspace*{-3mm}
		\end{figure}
		
		\begin{figure}
			\centering
			\includegraphics[width=0.35\columnwidth]{./figures/monotone_1_1.pdf}~~~~~~~\hfil
			\includegraphics[width=0.35\columnwidth]{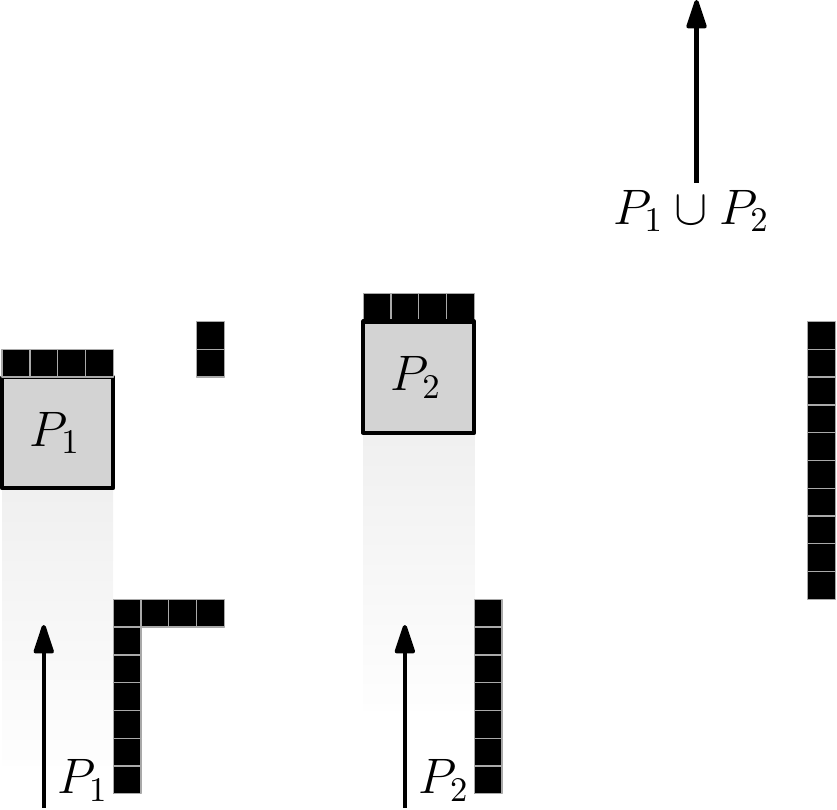}
			\caption{Gadgets assembling two subpolyominoes. Left: With unnecessary obstacles. Right: Without unnecessary obstacles.}
			\label{fig:gadget_cropped}
			\vspace*{-5mm}
		\end{figure}		

%% file: 05-demonstration.tex
\section{Experimental demonstration}\label{sec:demonstration}
\todo{...the experiment result
	lacks of an extensive comparison and quantitative analysis.}

	\begin{figure}
		\centering
		\includegraphics[width=\columnwidth]{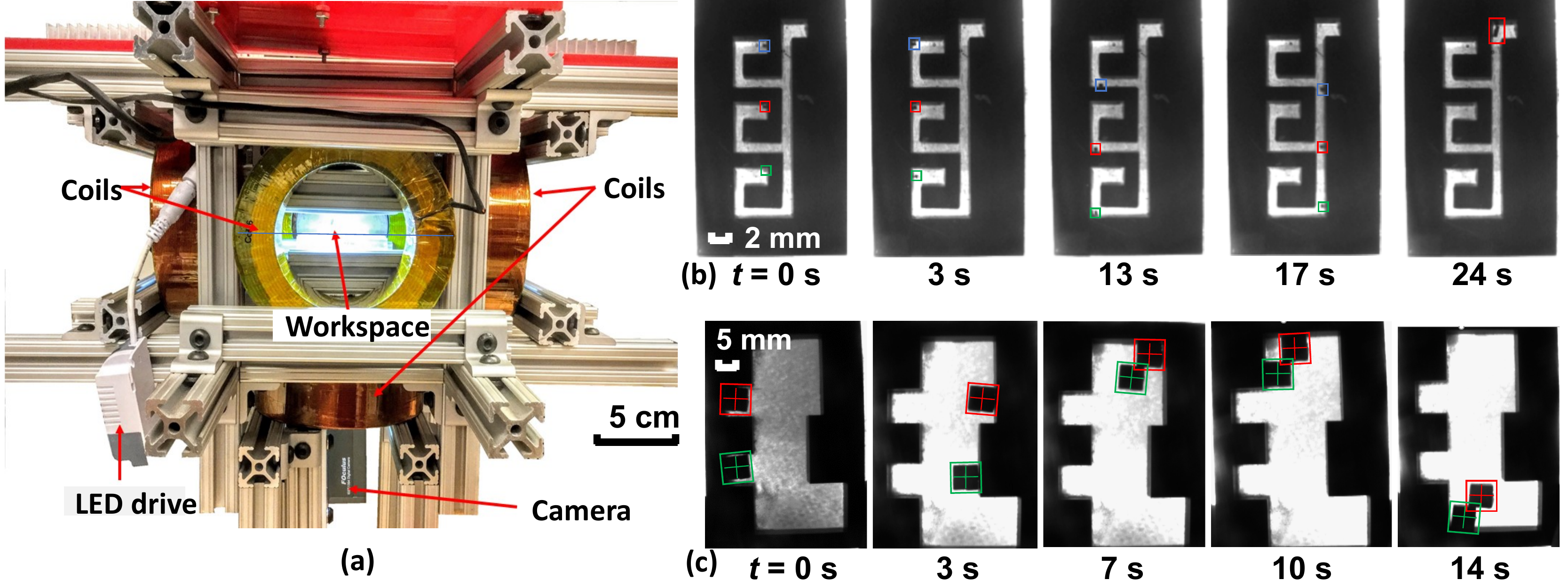}
		\caption{(a) Magnetic manipulation workspace
			(b) frames from an assembly of one column of a polyomino.
			(c) frames from combining two polyominoes.\label{fig:MagSetup}}		
	\end{figure}
	\begin{figure}
	\centering
	\includegraphics[width=\columnwidth]{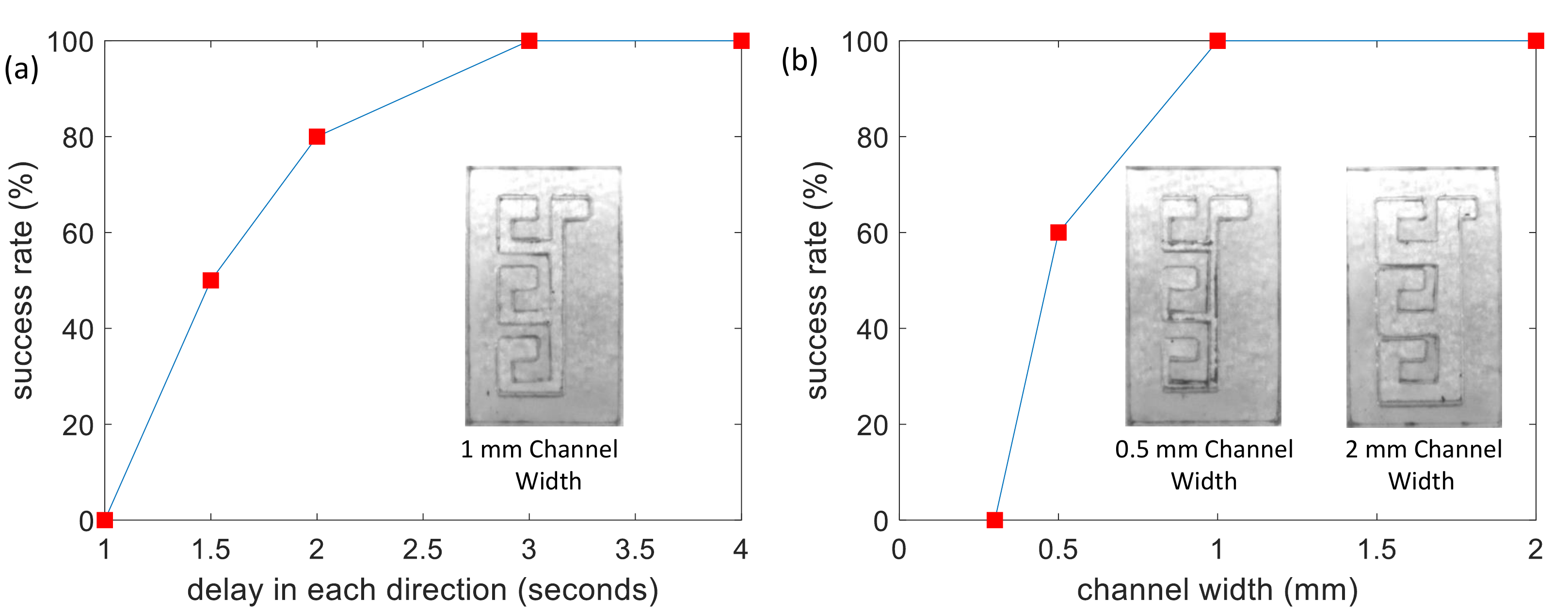}
	\vspace{-0.7cm}
	\caption{\revision{Results from assembly of a micro-scale three-tile column polyomino. There are 10 trials per data point.
		(a) Success rate as a function of the duration control inputs were applied in each of the four directions on a workspace with 1 mm width channels.
		(b) Success rate as a function of channel widths using control inputs applied for 3 s in each direction.}\label{fig:ColumnPlots}}
\vspace{-0.7cm}
\end{figure}
	
We  implemented the algorithms for staged assembly at micro and milli scale. A customized setup  was used to generate a magnetic field to manipulate the magnetic particles.
\paragraph{Experimental Platform}
\todo{The key point is to analyze the robustness of the algorithm if the
	size of particles and manufacturing/obstacles environment are
	inaccurate, because it seems there is a big assumption that the
	obstacles/particles/ environment setup should be accurate , and the
	question is how much it is practical.}
The magnetic setup used for the experiments is shown in Fig.~\ref{fig:MagSetup}, consisting of three orthogonal pairs of coils with separation distance equivalent to the outer diameter (127.5 mm) of a coil. 
 The coils (18 AWG, 1200 turns, Custom Coils, Inc) are actuated by six SyRen10-25 motor drivers, and a Tekpower HY3020E is used for the DC power supply. 
  The electromagnetic platform can provide \revision{uniform magnetic fields of up to 101 G, and gradient fields up to 150 mT/m} along any horizontal direction in the center of the workspace. 
   With flux concentration cores, up to 900 mT/m gradient fields are observed in the experiment. 
   \revision{Each flux concentration core is a  solid iron cylinder 73.1 mm in diameter.}

The workspaces used to demonstrate the sublinear assembly algorithms were designed to replicate  the column assembly in Fig.~\ref{fig:ortho_convex} and the subpolyomino assembly in Fig.~\ref{fig:monotone:p1p2}. Each workspace is made up of two layers of acrylic cut using a Universal Laser Cutter. The base layer is fabricated from 2 mm thick transparent acrylic, and it is glued to 5.5 mm thick  acrylic, which acts as an obstacle layout. In each experiment, the workspace is placed in the center of our electromagnetic platform. The particle tiles are composed of nickel-plated neodymium cube-shaped magnets (supermagnetman.com C0010). The magnet cubes have edge lengths of  0.5 mm for micro-scale and 2.88 mm for milli-scale demonstrations. An Arduino Mega 2560 was used to control the current in the coils and the workspaces were observed with a IEEE 1394 camera, captured at 60 fps. 
%The field of view was 75.7 mm $\times$ 51.7 mm and 37.4 mm $\times$ 37.3 mm for milli-scale and micro-scale respectively.

\paragraph{Experimental Results}
\todo{The simulation/experiment shows interesting result, however, there is
	no quantitative result. }
\revision{In micro-scale experiments, we filled the workspaces with vegetable oil and placed a magnet cube with 0.5 mm edge length in each of the three hoppers. 
 The workspace used in these experiments was 18 mm wide and 30 mm long. To assemble the column polyomino, a gradient magnetic field of 900 mT/m was applied in the direction sequence $ \langle d,r,u,l \rangle$. Each direction input was applied for a fixed amount of time specified by a {\sc Matlab} program. 
A successful trial requires that all three components are joined and delivered to the top right of the workspace.
 Fig.~\ref{fig:MagSetup}b shows the completed three-tile polyomino and Fig.~\ref{fig:ColumnPlots} shows representative experimental results for the assembly of the column polyomino. 
 Successful assembly depends on the channel widths and the duration of the control inputs.
 Larger channel widths and longer control durations led to high success rates.  
 Trials were always successful when the magnetic field was applied at least 3 s  in each direction and when the channel width was at least 1 mm. }

For milli-scale demonstrations we assembled two polyominoes, as shown in Fig.~\ref{fig:MagSetup}c. Each polyomino is composed of four magnet cubes glued together to form a square shape. The  43 mm $\times$ 62 mm workspace was placed in \revision{a uniform, 101 G magnetic 
field} to control the orientation of the polyominoes and then \revision{manually} tilted in the direction sequence $ \langle u,l,d,r,u \rangle$. See video attachment for experimental demonstrations.

%% file: 06-conclusion.tex
\section{Conclusion and Future Work}
\todo{More detail about the proposed algorithm with the
	summary of the algorithm would go to the conclusion. For example,  it
	would have a better flow if the detail on O() analysis from
	introduction
	moves to conclusion. Instead, the first sentence in conclusion can be
	moved to introduction. Generally speaking, the introduction usually
	describes the big picture of the problem statement, challenges and the
	proposed algorithm, and conclusion summarizes the algorithm and result
	with more technical detail. }
%We showed that we can decide if we are able to construct a given polyomino  $P$ with 2-cuts in polynomial time.
%Interestingly, the production time may not depend on the number of tiles $N$ of $P$, e.g., if the number of \revision{locally} reflex tiles is low.
%In case of $\Omega(N)$ \revision{locally} reflex tiles there are polyominoes that need $O(N)$ movement steps. 
% THIS EXAMPLE WAS NOT GIVEN?
% However, the example given can be produced in sublinear many steps if we are allowed to scale the polyomino.
%This raises one question for future work: 
%As we saw in the previous section there is still a lower bound of $\Omega(\sqrt N)$.
\revision{
A spectrum of future work remains, most notably issues of robustness in the presence of inaccuracies, as well 
as the extension of our results
to three-dimensional shapes. Questions in 2D include the following. 
Can we guarantee sublinear production times if the polyomino can be scaled by a constant?
Are straight cuts sufficient, i.e., if a polyomino $P$ is 2-cuttable, is $P$ also straight 2-cuttable?
 How hard is it to decide if a polyomino cannot be built at all?
Can we efficiently assemble polyomino $P'$ that approximates $P$?}
\todo{More reference (see reviewers suggestions)}